\documentclass[12pt]{iopart}

\usepackage{iopams}  
\expandafter\let\csname equation*\endcsname\relax

\expandafter\let\csname endequation*\endcsname\relax

\usepackage{amsmath}
\usepackage{enumitem}%

\usepackage{bm}
\usepackage{algorithm2e}
\usepackage{amsfonts}
\usepackage{graphicx}

\let\^=\hat
\let\'=\vec
\let\-=\mathbf
\def\~#1{{\mbox{\sf#1}}}
\def\E{{\mathbb E}}

\newcommand{\N}{{\mathbb{N}}}

\newcommand{\R}{{\mathbb{R}}}

\let\<=\langle
\let\>=\rangle

\def\pr{\mathrm{pr}}

\newtheorem{theorem}{Theorem}[section]
\newtheorem{lemma}[theorem]{Lemma}
\newtheorem{proposition}[theorem]{Proposition}
\newtheorem{corollary}[theorem]{Corollary}
\newenvironment{proof}[1][Proof]{\begin{trivlist}
\item[\hskip \labelsep {\bfseries #1}]}{\end{trivlist}}
\newtheorem{remark}[theorem]{Remark}

\begin{document}

\title[TV-Gaussian prior for Bayesian inverse problems]{A TV-Gaussian prior for infinite-dimensional Bayesian inverse problems and its numerical implementations}

\author{Zhewei Yao, Zixi Hu}

\address{Department of Mathematics and Zhiyuan College, Shanghai Jiao Tong University, Shanghai 200240, China}
\ead{zyaosjtu@gmail.com, hzx@sjtu.edu.cn}

\author{Jinglai Li}

\address{Institute of Natural Sciences, Department of Mathematics and MOE Key Laboratory of Scientific and Engineering Computing, Shanghai Jiao Tong University, Shanghai 200240, China}
\ead{jinglaili@sjtu.edu.cn}

\vspace{10pt}
\begin{indented}
\item[]January 2016
\end{indented}

\begin{abstract}
Many scientific and engineering problems require to perform Bayesian inferences in function spaces, where the unknowns are of infinite dimension. 
In such problems, choosing an appropriate prior distribution is an important task. 
In particular, when the function to infer is subject to sharp jumps, the commonly used Gaussian measures become unsuitable.
On the other hand, the so-called total variation (TV) prior  
can only be defined in a finite dimensional setting,   
and does not lead to a well-defined posterior measure in function spaces. 
In this work we present a TV-Gaussian (TG) prior to address such problems, where the TV term is used to detect sharp jumps of the function,
and the Gaussian distribution is used as a reference measure so that it results in a well-defined posterior measure in the function space.
We also present an efficient Markov Chain Monte Carlo (MCMC) algorithm to draw samples from the posterior distribution of the TG prior. 
With numerical examples we demonstrate the performance of the TG prior and the efficiency of the proposed MCMC algorithm.

\end{abstract}

%
\vspace{2pc}
\noindent{\it Keywords}: Bayesian inference, Gaussian measure, Markov Chain Monte Carlo, total variation.
%
%
%
%

\section{Introduction}
\label{sec:intro}

The Bayesian inference methods~\cite{gelman2014bayesian} for solving inverse problems have gained increasing popularity, largely due to their ability to quantify uncertainties in the estimation results.
The Bayesian inverse problems have been extensively studied in the finite dimensional setting~\cite{aster2013parameter,kaipio2005statistical},
and more recently, a rigorous Bayesian framework~\cite{stuart2010inverse} is developed for the inverse problems in function spaces where the unknowns are of infinite dimension. 
In existing works that perform Bayesian inferences in function spaces, Gaussian measures are widely used 
as the prior distributions. 
Such a choice is in fact well justified as theoretical studies suggest that the Gaussian measures are
 well-behaved priors:  the resulting posterior depends continuously on the data and
it can be well approximated by finite dimensional representations. 
However, in many practical problems such as image reconstructions, the true functions that one aims to infer are often subject to sharp jumps or even discontinues,
which can not be well modeled by Gaussian priors.  
In the deterministic inverse problem context, such functions are often estimated 
using the total variation (TV) regularization~\cite{rudin1992nonlinear}. 
To this end,  it is very intriguing to perform Bayesian inferences of the functions with sharp jumps, with a TV prior, the construction of which is rather straightforward in the finite dimensional setting. 
Thus to use the prior in the infinite dimensional setting, one first represents the unknown function with a finite-dimensional parametrization, for example, by discretizing the function on a pre-determined mesh grid, and then solve the resulting finite dimensional inference problem with the TV prior. 
Such a method has been successfully applied to a variety of problems~\cite{vogel1998fast,babacan2009variational}. 
A major issue of the TV prior is that, unlike the finite-dimensional Gaussian distribution, the posterior distribution of the TV prior may not converge 
to a well-defined infinite dimensional measure as the discretization dimension increases, a property that is referred to as being 
discretization variant in \cite{lassas2004can}. 
In other word, the inference results depend on the discretization dimensions, which is certainly undesirable from both practical and theoretical points of view.  
A number of non-Gaussian priors therefore have been proposed to address the issue.
For example, the work~\cite{lassas2009discretization} proposes to use the Besov priors based on wavelet expansions in such problems
 and the theoretical properties of which are further investigated in \cite{dashti2012besov}.
In another work~\cite{helin2011hierarchical}, a hierarchical Gaussian prior related to the Mumford-Shah regularization in the deterministic setting is developed. 

 We note that the priors mentioned above differ significantly from the Gaussian measures, in both theoretical properties and numerical implementations.  
As a result, many analysis techniques and numerical methods developed for the Gaussian priors can not be easily extended to these non-Gaussian priors.  
In particular, as will be discussed later, some Markov Chain Monte Carlo (MCMC) algorithms developed for Gaussian priors may not be applied directly to, for example, the Besov ones. 
To this end, efforts have been made to developing particular MCMC algorithms for those non-Gaussian priors, e.g.~\cite{vollmer2015dimension}.
As an alternative solution, in this work we propose a TV-Gaussian (TG) prior which is motivated by the elastic net 
regularization for linear regression problems~\cite{zou2005regularization},
and the hybrid regularization method in image problems~\cite{compton2012hybrid}. 
Namely, the prior includes a TV term to detect edges, and on the other hand, 
it uses Gaussian distributions as a reference measure, so that it leads to a well-defined posterior distribution in the function space.

MCMC simulations are widely used to draw samples from the posterior distribution in Bayesian inferences.
It has been known that many standard MCMC algorithms, 
such as the random walk Metropolis-Hastings,  can become arbitrarily slow as the discretization mesh of the unknown is refined~\cite{roberts2001optimal,cotter2013mcmc}. 
Namely the mixing time of an algorithm can increase to infinity as the dimension of the discretized parameter approaches to infinity,
and in this case the algorithm is said to be \emph{dimension-dependent}.
A family of dimension-independent MCMC algorithms were presented in \cite{cotter2013mcmc} by constructing a preconditioned Crank-Nicolson (pCN)
discretization of a stochastic partial differential equation that preserves the reference measure. 
There are a number of other algorithms for the infinite dimensional problems that can further improve the sampling efficiency by incorporating the data information: 
 the stochastic Newton MCMC~\cite{martin2012stochastic},
the dimension-independent likelihood-informed MCMC~\cite{cui2014dimension}, and the adaptive independence 
sampler algorithm developed in \cite{feng2015adaptive},  just to name a few.  
All the aforementioned algorithms are developed based on Gaussian priors,  
but thanks to the Gaussian reference measure, many of these algorithms, most notably the pCN method,  can be directly applied to
problems with our TG priors (some gradient based algorithms may require certain modifications). 
Moreover, by taking advantage of the special structure of the TG prior, we propose a simple splitting scheme to accelerate the pCN algorithm. 
Loosely speaking, the splitting scheme has two stages: in the first stage the sample is moved (accepted/rejected) several times 
only according to the TV term, and in the second stage it is finally accepted or rejected according to the likelihood function.  
We prove the detailed balance of the proposed splitting pCN (S-pCN) scheme and with numerical examples we demonstrate that
the method can significantly improve the mixing rate by making very simple modifications to the standard pCN method. 

To summarize, the main contributions of the work are two-fold: we propose a TV-Gaussian hybrid prior to handle unknowns that can not be well modeled
by standard Gaussian measures, and we also provide an efficient MCMC algorithm specifically designed for the proposed prior. 
The rest of the work is organized as the following. In section~\ref{sec:prior}, we present the TG priors and 
provide some results regarding the theoretical properties of it.  In section~\ref{sec:spcn}, we describe the S-pCN algorithm
to efficiently sample from the proposed TG prior.
 Numerical examples are provided in section~\ref{sec:examples} to demonstrate the 
performance of the TG prior and the efficiency of the S-pCN algorithm. 
Finally section~\ref{sec:conclusions} offers some concluding remarks.

\section{The TV-Gaussian  priors}\label{sec:prior}
We describe the TG priors in this section, starting by a general introduction of Bayesian inverse problems in function spaces. 

\subsection{Problem setup}\label{s:setup}
We consider a separable {Hilbert} space $X$ with inner product $\<\cdot,\cdot\>_X$.
 Our goal is to estimate the unknown  $u\in X$ from data $y\in \R^{m}$. The data $y$ is related to $u$ %
via a forward model,
\begin{equation}
y = G(u)+\zeta,
\end{equation}
where $G:X\rightarrow \R^{m}$ and
$\zeta$ is a $m$-dimensional zero mean Gaussian noise with covariance matrix $\Sigma$.
In particular we assume that the data $y$ is generated by applying the operator $G$ to a truth $u^\dagger\in X$ and then adding noise to it.
It is easy to see that, under this assumption, the likelihood function, i.e., the distribution of $y$ conditional on $u$ is,
\[p(y|u) \propto \exp(-\Phi^y(u)),\]
where 
\begin{equation}
\Phi^y(u) := \frac12\|G(u)-y\|^2_\Sigma=\frac12\|\Sigma^{-1/2}(G(u)-y)\|_2^2,
\end{equation}
is often referred to as the data fidelity term in deterministic inverse problems. 
In what follows, without causing any ambiguity, we shall drop the superscript $y$ in $\Phi^y$ for simplicity. 
In the Bayesian inference we assume that the prior measure of $u$ is $\mu_\pr$, and the posterior measure of $u$ 
is provided by the Radon-Nikodym (R-N) derivative:
\begin{equation} \frac{d\mu^y}{d\mu_{\pr}}(u) =\frac1Z\exp(-\Phi(u)), \label{e:bayes}
\end{equation}
where $Z$ is a normalization constant.
Eq.~\eqref{e:bayes}  can be interpreted as the Bayes' rule in the infinite dimensional setting.


%
As is mentioned in the Section~\ref{sec:intro}, it is conventionally assumed that the prior $\mu_\pr=\mu_0$ where $\mu_0 = N(0,C_0)$, i.e., a  Gaussian measure defined on $X$ with (without loss of generality) zero mean  and covariance operator $C_0$.
Note that $C_0$ is symmetric positive and of trace class.
The range of $C_0^{\frac12}$,
\[E = \{u = C_0^{\frac12} x\, |\, x\in X\}\subset X,\]
which is a Hilbert space equipped with inner product~\cite{da2006introduction},
\[\<\cdot,\cdot\>_E = \<C_0^{-\frac12}\cdot,C_0^{-\frac12}\cdot\>_X ,\]
is called the Cameron-Martin space of measure $\mu_0$. 
For the Gaussian prior, it has been proved that if $G$ satisfies the following assumptions~\cite{stuart2010inverse}:\\
\noindent\textbf{Assumptions A.1.}
\begin{enumerate}[label=\roman*]
\item for every $\epsilon>0$ there is $M=M(\epsilon)\in \R$ such that, for all $u\in X$,
\[\|G(u)\|_\Sigma \leq \exp(\epsilon\|u\|_X^2+M),
\]
\item for every $r>0$ there is $K=K(r)>0$ such that, for all $u_1,\,u_2 \in X$ with $\max\{\|u_1\|_X,\,\|u_2\|_X\}<r$,
\[\|G(u_1)-G(u_2)\|_\Sigma \leq K \|u_1-u_2\|_X,\]
\end{enumerate}
the associated functional $\Phi$ satisfies Assumptions 2.6 in \cite{stuart2010inverse}.
As a result, the posterior $\mu^y$ is a well-defined probability measure on $X$ (Theorem 4.1 in \cite{stuart2010inverse})  and it 
is Lipschitz in the data $y$ (Theorem 4.2 in \cite{stuart2010inverse}).
Moreover, under some additional assumptions, the posterior measure can be well approximated by a finite dimensional representation.

We also should note that the definition of the maximum a posteriori (MAP) estimator in finite dimensional spaces does not apply here, as the measures $\mu^y$ and $\mu_0$ are not absolutely
continuous with respect to the Lebesgue measure; instead,
the MAP estimator in $X$ is defined as the minimizer of the Onsager-Machlup functional (OMF)~\cite{dashti_map_2013,Li20151}:
\begin{equation}
 I(u) := \Phi(u)+\frac12\|u\|^2_E,\label{e:OM}
\end{equation}
over the Cameron-Martin space $E$ of $\mu_0$. Here the Cameron-Martin norm $\|\cdot\|_E$ is given by
 \[\|u\|_E = \|C_0^{-\frac12} u\|_X.\]

\subsection{A general class of the hybrid priors}
We present a general class of hybrid priors in this section. 
The idea is rather straightforward: instead of simply letting $\mu_\pr=\mu_0$, we let 
\begin{equation} 
\frac{d\mu_\pr}{d\mu_{0}}(u) \propto\exp(-R(u)), \label{e:prior}
\end{equation}
where $R(u)$ represents additional prior information (or regularization) on $u$. 
In what follows we shall refer to $R$ as the additional regularization term. 
It follows immediately that the R-N derivative of $\mu^y$ with respect to $\mu_0$ is
\begin{equation} 
\frac{d\mu^y}{d\mu_0}(u) \propto\exp(-\Phi(u)-R(u)), \label{e:bayes2}
\end{equation}
which returns to the conventional formulation with Gaussian priors.
Next we shall show that Eq.~\eqref{e:prior} is a well-behaved prior under certain assumptions on $R$:

\noindent\textbf{Assumptions A.2.} The function $R: X \to\mathbb{R}$ has the following properties.
\begin{enumerate}[label=\roman*]
\item For all $u\in X$, $R(u)$ is bounded from below, and without loss of generality we can simply assume
$R(u)\geq0$.
\item For every $r>0$ there is a $K=K(r)>0$ such that, for all $u\in X$ with $\Vert u\Vert_X<r$,
$R(u)\leq K$.

\item For every $r>0$, there is an $L(r)>0$ such that, for all $u_1,u_2\in X$  with $\max\lbrace\Vert u_1\Vert_X,\Vert u_2\Vert_X\rbrace<r$,
\[
| R(u_1)-R(u_2)|\leq L\Vert u_1-u_2\Vert_X.
\]
\end{enumerate}
Note that the assumptions above are similar to the Assumptions~(2.6) in \cite{da2006introduction}, 
with two major differences: first in A.1 (i), we require $R$ to be strictly bounded from below; secondly $R$ is independent of the data $y$ and so we do not have item (iv) in Assumption 2.6.
The requirement that $R$ is bounded from below is needed in the proof of our results regarding the finite-dimensional approximation of the posterior. 
It is easy to show that if $\Phi$ satisfy Assumptions 2.6 in and $R$ satisfy Assumptions~A.1, $\Phi+R$ satisfies Assumption~2.6 in~\cite{da2006introduction}.
As a result, $\mu^y$ is a well-defined measure on $X$ and it is Lipschitz in the data $y$, which
are stated in the following two theorems.
\begin{theorem}
Let $G$ satisfy Assumptions A.1 and $R$ satisfy Assumptions~A.2. Then $\mu^y$ given by Eq.~\eqref{e:bayes2} is a well-defined probability measure on $X$. \label{th:1}
\end{theorem}

\begin{theorem} Let $G$ satisfy Assumptions A.1 (i) and $R$ satisfy Assumption A.2 (i)-(ii).
Then $\mu^y$ given by Eq.~\eqref{e:bayes2} is Lipschitz in the data $y$, with respect to the Hellinger distance: if $\mu^y$ and $\mu^{y'}$ are two measures corresponding to data $y$ and $y'$ the there exists $C=C(r)$ such that, for all $y,y'$ with $\max\lbrace\Vert y\Vert_2,\Vert y'\Vert_2\rbrace<r$,
\[
d_\mathrm{Hell}(\mu^y,\mu^{y'})\leq C\Vert y-y'\Vert_\Sigma.
\]
Consequently the expectation of any polynomially bounded function $f:X\to E$ is continuous in $y$. \label{th:2}
\end{theorem}

Theorems \eqref{th:1} and \eqref{th:2} are direct consequences of the fact that $\Phi+R$ satisfies Assumption~2.6 in \cite{stuart2010inverse} and so we omit the proofs here. 
Next we shall study the related issue of approximating the posterior $\mu^y$ with a measure defined in a finite dimensional space, which is of essential importance for numerical implementations of the Bayesian inferences.  
In particular we consider the following approximation: 
\begin{equation}
\frac{d\mu_{N_1,N_2}^y}{d\mu_0} = \exp (-\Phi_{N_1}(u)-R_{N_2}(u)),
\end{equation}
where $\Phi_{N_1}(u)$ is a $N_1$ dimensional approximation of $\Phi(u)$ and $R_{N_2}(u)$ is a $N_2$ dimensional approximation of $R(u)$. 
The following theorem provides the convergence of $\mu^y_{N_1,N_2}$ to $\mu^y$ with respect to Hellinger distance under certain assumptions. 
\begin{theorem} \label{th:3}
Assume that $G$ and $G_{N_1}$ satisfy Assumption A.1 (i) with constants uniform in $N_1$, and $R$ and $R_{N_2}$ satisfy Assumptions A.2 (\romannumeral1) and (\romannumeral2) 
 with constants uniform in $N_2$. Assume also that for $\forall\epsilon>0$, there exist two positive sequences $\{a_{N_1}(\epsilon)\}$ and $\{b_{N_2}(\epsilon)\}$ converging to zero, such that
$\mu_0(X_\epsilon)\geq1-\epsilon$ for $\forall\, N_1,N_2\in \N$, 
where 
\[X_\epsilon 
= \{u\in X\, |\,\vert\Phi(u)-\Phi_{N_1}(u)\vert\leq a_{N_1}(\epsilon),\,
\vert R(u)-R_{N_2}(u)\vert\leq b_{N_2}(\epsilon) \}.
\]
Then we have
$$
d_\mathrm{Hell}(\mu^y,\mu^y_{N_1,N_2})\to0 ~~~~\mathrm{as}~~~~ N_1,\,N_2\to+\infty.
$$
\end{theorem}
\begin{proof}
For every $r>0$ there is a $K_1=K_1(r)>0$ and a $K_2=K_2(r)>0$ such that, for all $u\in X$ with $\Vert u\Vert_X<r$,
$\Phi(u)\leq K_1$ and $R(u)\leq K_2$.
Let $r=\|y\|_\Sigma$ and $K(r)=K_1(r)+K_2(r)$, and we can show that the normalization constant $Z$ for $\mu^y$ satisfies, 
\[
Z\geq\int_{\lbrace\Vert u\Vert_X<r\rbrace}\exp(-K(r))\mu_0(\mathrm{d}u)=\exp(-K(r))\mu_0\lbrace\Vert u\Vert_X<r\rbrace=C.
\]
Similarly, we can show that $Z_{N_1,N_2}\geq C$ where $Z_{N1,N_2}$ is the normalization constant for $\mu^y_{N_1,N_2}$.
Since for any $a>0$ and $b>0$, $\vert\exp(-a)-\exp(-b)\vert\leq\min\lbrace1,\vert a-b\vert\rbrace$, for any $\epsilon\in(0,1)$,
\begin{align*}
\vert Z-Z_{N_1,N_2}\vert &\leq \int_X\vert\exp(-\Phi(u)-R(u))-\exp(-\Phi_{N_1}(u)-R_{N_2}(u))\vert\mathrm{d}\mu_0(u)\\
&\leq\int_{X\backslash X_\epsilon}\mu_0(\mathrm{d}u)+\int_{X_\epsilon}\vert\Phi(u)-\Phi_{N_1}(u)\vert\mu_0(\mathrm{d}u)\\
&+\int_{X_\epsilon}\vert R(u)-R_{N_2}(u)\vert\mu_0(\mathrm{d}u)\\
&\leq\epsilon+a_{N_1}(\epsilon)+b_{N_2}(\epsilon).
\end{align*}

From the definition of Hellinger distance, we have
\begin{multline}
2d_\mathrm{Hell}(\mu^y,\mu^y_{N_1,N_2})^2 = \int_X \left(\sqrt{\frac{d\mu^y}{d\mu_0}}-\sqrt{\frac{d\mu^y_{N_1,N_2}}{d\mu_0}}\right)^2 \mu_0(du)
\\=\int_X(Z^{-\frac12}\exp(-\frac12\Phi(u)-\frac12R(u))-Z_{N_1,N_2}^{-\frac12}\exp(-\frac12\Phi_{N_1}(u)-\frac12R_{N_2}(u)))^2\mu_0(\mathrm{d}u)\\
\leq I_1+I_2+I_3,
\end{multline}
where
\[
I_1=\int_{X\backslash X_\epsilon}(\frac1{\sqrt{Z}}\exp(-\frac12\Phi(u)-\frac12R(u))-\frac1{\sqrt{Z_{N_1,N_2}}}\exp(-\frac12\Phi_{N_1}(u)-\frac12R_{N_2}(u)))^2\mu_0(\mathrm{d}u),
\]
\[
I_2=\frac2Z\int_{X_\epsilon}(\exp(-\frac12\Phi(u)-\frac12R(u))-\exp(-\frac12\Phi_{N_1}(u)-\frac12R_{N_2}(u)))^2\mu_0(\mathrm{d}u),
\]
\[
I_3=2\vert Z^{-\frac12}-(Z_{N_1,N_2})^{-\frac12}\vert^2\int_{X_\epsilon}\exp(-\Phi_{N_1}(u)-R_{N_2}(u))\mu_0(\mathrm{d}u).
\]
Actually, it is easy to show
\[
I_1\leq\int_{X\backslash X_\epsilon}(2C^{-\frac12})^2\mu_0(\mathrm{d}u)\leq C\epsilon,
\]
\[
I_2\leq\frac2C\int_{X_\epsilon}(a_{N_1}(\epsilon)+b_{N_2}(\epsilon))^2\mu_0(\mathrm{d}u)\leq C(a_{N_1}(\epsilon)+b_{N_2}(\epsilon))^2,
\]
\[
I_3\leq C(Z^{-3}\wedge(Z_{N_1,N_2})^{-3})\vert Z-Z_{N_1,N_2}\vert^2\int_{X_\epsilon}\mu_0(\mathrm{d}u)=C(\epsilon+a_{N_1}(\epsilon)+b_{N_2}(\epsilon))^2.
\]
It follows immediately that
\[
2d^2_\mathrm{Hell}(\mu^y,\mu^y_{N_1,N_2})\leq C(\epsilon+\epsilon^2+(a_{N_1}(\epsilon)+b_{N_2}(\epsilon))^2+\epsilon(a_{N_1}(\epsilon)+b_{N_2}(\epsilon))),
\]
where $C$ is a constant independent of $N_1,N_2$. Let $N_1$ and $N_2$ tend to $+\infty$, yielding
\[
\lim_{N_1,N_2\to+\infty}2d_\mathrm{Hell}(\mu,\mu_{N_1,N_2})^2\leq C(\epsilon+\epsilon^2),
\]
for any $\epsilon>0$. Thus,
\[
\lim_{N_1,N_2\to+\infty}d_\mathrm{Hell}(\mu^y,\mu^y_{N_1,N_2})=0,
\]
which completes the proof. 
\end{proof}

We emphasize that the major difference 
between Theorem \ref{th:3} and Theorem 4.10 in \cite{stuart2010inverse} is that our assumption is weaker than 
that in Theorem 4.10. The reason that we can use a weaker assumption is that we require $R$ to be strictly bounded from below in Assumption A.2 (i). 
This modification of assumptions is important for our work as that the key assumption made in Theorem 4.10 may not hold in our setting~(see Remark \ref{rem:assum} for details).

In general, the approximation errors arise from two sources: representing $u$ with a finite dimensional basis,
and solving the forward model $G$ (or equivalently the functional $\Phi$) approximately.  
Next we consider a special case where we assume that for a given finite dimensional representation $u_N$ of $u$, the functional $\Phi(u_N)$ can be computed exactly; this can be understood as the idealized formulation
 of the situation where for any given $u_N$ one can choose a numerical scheme to compute the solution to a desired level of accuracy.  
In this setting, we can show that the finite dimensional approximation $\mu^y_N$ converges to $\mu$ without assuming any additional conditions:  
\begin{corollary} \label{cor:1}
Let $\{e_k\}_{k=1}^\infty$ be a complete orthonormal basis of $X$, 
\begin{equation}
u_N=\sum_{k=1}^N\<u,e_k\> e_k, \label{e:u_N}
\end{equation}
and  
\[
\frac{d\mu^y_N}{\mu_0}=\exp(-\Phi(u_N)-R(u_N)).
\]
Assume that $G$ satisfies Assumptions~A.1 and $R$ satisfies Assumptions A.2. Then 
\[
d_\mathrm{Hell}(\mu^y,\mu^y_N)\to0,~~\mathrm{as}~~N\to\infty.
\]
\end{corollary}
\begin{proof}
 Set
\[
a_N=\E\Vert u-u_N\Vert_X^2= \sum_{k=N+1}^\infty \E|\<u,e_k\>|^2. 
\]
Since $C_0$ is in the trace class, $a_N\to 0$ as $N\to\infty.$
By Markov's inequality, we have that, for any $\epsilon>0$, 
\begin{equation}
\mu_0(\lbrace\Vert u-u_N\Vert_X>\sqrt{\frac{2a_N}{\epsilon}}\rbrace)\leq \frac12\epsilon,\quad \mathrm{for}\, \forall\, N\in \N. \label{e:eq1}
\end{equation}

For the given $\epsilon$, there is a $r_\epsilon$ such that $\mu_0(\{u\in X\, |\, \Vert u\Vert_X>r_\epsilon\})<\frac12\epsilon.$ 
It is easy to show that, for any $N\in\N$, 
\[
\mu_0(\lbrace u\in X\, |\, \Vert u\Vert_X\leq r_\epsilon,\,\Vert u-u_N\Vert_X\leq\sqrt{\frac{2a_N}{\epsilon}}\rbrace)\geq 1-\epsilon.
\]
For conciseness, we define $\widetilde{X}=\lbrace u\in X\, |\, \Vert u\Vert_X\leq r_\epsilon,\,\Vert u-u_N\Vert_X\leq\sqrt{\frac{2a_N}{\epsilon}}\rbrace.$ 

Recall that $\Phi$ satisfies Assumptions 2.6 in \cite{da2006introduction} and $R$ satisfies Assumptions A.2, and so 
we know that
there are constants $L^\Phi_\epsilon, L^R_\epsilon>0$, such that for any $u\in \widetilde{X}$,
\begin{gather*}
\vert\Phi(u)-\Phi(u_N)\vert\leq L^\Phi_\epsilon\Vert u-u_N\Vert_X\leq L^\Phi_\epsilon\sqrt{\frac{2a_N}{\epsilon}},\\
\vert R(u)-R(u_N)\vert\leq L^R_\epsilon\Vert u-u_N\Vert_X\leq L^R_\epsilon\sqrt{\frac{2a_N}{\epsilon}}.
\end{gather*}
Obviously, $L^\Phi_\epsilon\sqrt{\frac{2a_N}{\epsilon}}, L^R_\epsilon\sqrt{\frac{2a_N}{\epsilon}}\to0$ as $N\to\infty$. 
As,
 \[\widetilde{X}\subset X_\epsilon=\{u\in X\, |\,\vert\Phi(u)-\Phi(u_N)\vert\leq L^\Phi_\epsilon\sqrt{\frac{2a_N}{\epsilon}}, \,
\vert R(u)-R(u_N)\vert\leq L^R_\epsilon\sqrt{\frac{2a_N}{\epsilon}}\},\]
we have $\mu_0(X_\epsilon)\geq 1-\epsilon$,
and by Theorem \ref{th:3}, 
\[
d_\mathrm{Hell}(\mu^y,\mu^y_N)\to0,~~\mathrm{as}~~N\to\infty,
\]
which completes the proof. 
\end{proof}

We reinstate that Theorem 4.10 in \cite{stuart2010inverse} can not be directly applied to this setting as its assumption may not necessarily be satisfied. 
Namely, to apply Theorem 4.10,  $\Phi+R$ must satisfy:
 for any $r>0$, there is a $K=K(r)>0$ such that for all $u$ with $\|u\|_X <r$, 
\[ |(\Phi(u)+R(u))-(\Phi(u_N)+R(u_N))| < K \phi(N),\]
where $\phi(N) \rightarrow 0$ as $N \rightarrow \infty$. 
In the following remark, we construct a simple counterexample for such a condition. 
\begin{remark} \label{rem:assum}
Let $u_N$ be given by Eq.~\eqref{e:u_N} and let $R$ satisfy $R(0)=0$.
Let $\Phi$ satisfy the assumption:
for any $r>0$, there is a $K=K(r)>0$ such that for all $u$ with $\|u\|_X <r$, 
\[ |\Phi(u)-\Phi(u_N)| < K \phi(N),\]
where $\phi(N) \rightarrow 0$ as $N \rightarrow \infty$. 
For any given $N$, we consider $u = e_{N+1}$ and obviously $u_N = 0$. It follows immediately that
\[  |(\Phi(u)+R(u))-(\Phi(u_N)+R(u_N))| > R(e_{N+1})-K\phi(N).\]

\end{remark}
It should also be noted that, in this work we only consider finite dimensional observation data for simplicity.
A more general case is to consider infinite dimensional data as well, and study
the limiting behavior of the posterior distributions under measurement refinements.  
The matter has been discussed in \cite{lassas2009discretization,dashti2012besov} for Besov priors. 

\subsection{The TV-Gaussian prior} 
Until this point, we have discussed the non-Gaussian priors in a rather general setting. 
In this section, we introduce the particular choice of the additional regularization term $R$ for the edge-detection purposes. 
First we need to specify the space of functions $X$. 
Let   $\Omega$  be a bounded open subset of  $\R^q$ where $q\in\N$, and $X$ be the Sobolev space $H^1(\Omega)$: 
\[X=H^1(\Omega) = \{ u \in L_2(\Omega)\, |\, \partial^\alpha_x u \in L_2(\Omega) \,\mathrm{for\,all}\, |\alpha|\leq1\, \},\] 
where where $\alpha=(\alpha_1,...,\alpha_q)$ and $|\alpha|=\sum_{i=1}^q \alpha_i$,
and the associated norm $\|\cdot\|_X=\|\cdot\|_{H^1}$ is 
\[\|u\|_{H^1} = \sum_{|\alpha|\leq1} \|\partial^\alpha_x u\|_{L^2(\Omega)}. \]
Naturally we choose the regularization term to be the TV seminorm,
\begin{equation}
R(u) = \lambda\| u\|_\mathrm{TV} = \lambda\int \|\nabla u\|_2 dx, \label{e:tv} 
\end{equation}
where $\lambda$ is a prescribed positive constant. 
We show that the TV term \eqref{e:tv} satisfies the required assumptions: 
\begin{lemma}
 Eq.~\eqref{e:tv} satisfies Assumptions~A.2. \label{lm:tv}
\end{lemma}

\begin{proof}
\noindent(\romannumeral1): The assumption is trivially satisfied. 

\noindent(\romannumeral2): For all $u\in X$ with $\Vert u\Vert_X<r$, there exists a constant $C>0$ such that 
\[R(u) =\lambda\Vert u\Vert_\mathrm{TV}
\leq \lambda C\Vert u\Vert_{X}
\leq \lambda Cr
=K(r).
\]

\noindent(\romannumeral3): For every $r>0$ and all $u_1,u_2\in X$, there is a constant $C>0$ such that  
\[
\vert R(u_1)-R(u_2)\vert
\leq \lambda\Vert (u_1-u_2)\Vert_\mathrm{TV}
\leq \lambda C\Vert u_1-u_2\Vert_X.\]
\end{proof}

It follows directly from Lemma~\ref{lm:tv} that Theorems \ref{th:1} and \ref{th:2} and Corollary~\ref{cor:1} hold for the prior~\eqref{e:tv}.
Following the same procedure of \cite{dashti_map_2013}, we can show that the MAP estimator for the TV-Gaussian prior is the minimizer of
\begin{equation}
 I(u) := \Phi(u)+\lambda \| u\|_\mathrm{TV}+\frac12\|u\|^2_E.\label{e:OM2}
\end{equation}

\section{Splitting pCN}\label{sec:spcn}
We now discuss the numerical implementation of the Bayesian inference with the TG priors. 
Typically the Bayesian inference is implemented with MCMC algorithms, 
and the pCN algorithms, a family of dimension-independent MCMC schemes, have been proposed to 
draw samples in the function spaces. 
As is mentioned in section \ref{sec:intro}, thanks to the special structure of the TG prior, the pCN algorithms can be applied directly to problems with this prior, requiring no modifications. 
Nevertheless, in this section, we propose to further improve the sampling efficiency by making very simple modifications to the pCN algorithms. 
The idea is based on the following two observations on our TG prior: 
\begin{enumerate}
\item The TV term $R$ is more sensitive to small local fluctuations of $u$ than the data fidelity term $\Phi$. 
\item The TV term $R$ can be computed much more efficiently than the data fidelity term $\Phi$.
In fact, evaluating $\Phi$ requires to simulate forward model $G$ which is often governed by some computationally intensive partial differential equations (PDE). 
\end{enumerate}
As a result of the first observation, to achieve a reasonable acceptance probability, for example, $20\%$, 
one has to use a very small stepsize in the pCN algorithm, resulting in poor mixing. 
On the other hand, the restriction of the stepsize is mainly due to the TV term. 
To address the issue, we introduce a splitting scheme to the standard pCN algorithm; namely, we perform the pCN in two stages: one for the TV term and one for the data fidelity term.   
The intuition behind the method is rather straightforward: the slowly varying and computationally intensive data fidelity term $\Phi$ is evaluated less frequently than the fast-varying
and computationally efficient TV term $R$. A particle is moved and accepted or rejected $k$ times
according to the TV term. The sample resulting from the $k$ short step
moves is accepted or rejected using a Metropolis criterion
based on $\Phi$  after the $k$ short-range moves.
Specifically, suppose the current state is $u_\mathrm{current}$, the splitting pCN (S-pCN) 
proceeds as follows: 
\begin{enumerate}
\item let $v_0 = u_\mathrm{current}$.
\item for $i=1$ to $k$ perform the following iteration:
\begin{enumerate}
\item propose   $v_\mathrm{prop} = \sqrt{1-\beta^2} v_{i-1} + \beta w$, where $w \sim \mu_0$; \label{st:a} 
\item 
let $v_i =v_\mathrm{prop}$ with probability; 
\begin{equation}
\mathrm{acc}_{R}(v_\mathrm{prop},v_{i-1}) = \min \{1, \exp[ -(R(v_\mathrm{prop}) - R(v_{i-1}))]\};\label{e:accR}
\end{equation} 
and let $v_i = v_{i-1}$ with probability $1-\mathrm{acc}_{R}(v_\mathrm{prop},v_{i-1}) $.
\item return to step \eqref{st:a};
\end{enumerate}
\item 
let $u_\mathrm{next}=v_k$ with probability 
\begin{equation}
\mathrm{acc}_{\Phi}(v_k,u_\mathrm{current})  = \min \{1, \exp[-( \Phi(v_k) - \Phi(u_\mathrm{current}))]\},\label{e:accphi}
\end{equation}
and $u_\mathrm{next} = u_\mathrm{current}$ with probability $1-\mathrm{acc}_{\Phi}(v,u_\mathrm{current})$.
\end{enumerate}
We note the similarity of this splitting algorithm and the approximation-accelerated MCMC algorithms such as \cite{christen2005markov,efendiev2006preconditioning}, and especially the surrogate transition method in \cite{liu2008monte}. 
However, our algorithm simply splits the data fidelity term and the TV term, without requiring any approximations or surrogate models. 
In fact, if a surrogate is available for the forward model, it can be naturally incorporated in to the S-pCN algorithm to further improve the sampling efficiency. 
Nevertheless, constructing and implementing such a surrogate model is not in the scope of this work.  

A possible interpretation of the S-pCN algorithm is that step 2 is responsible of generating a new proposal for the target distribution $\mu^y$,
which is rejected or accepted according to the normal MCMC rule in step 3. 
Under this interpretation,  we need to show that the detailed balance condition is satisfied, which is stated by the following proposition. 
\begin{proposition}
Let $q(u,dv)$ is the proposal kernel corresponding to the $v_k$ proposed in step 2 given $u$, and $Q$ is its transition kernel given by,
  \[
Q(u,dv)=a(u,v_k)q(u,dv)+\delta_u(dv)(1-\int \mathrm{acc}_\Phi(z,u)q(u,dz)),
\] 
where $\mathrm{acc}_\Phi(v_k,u)$ is given by Eq.~\eqref{e:accphi}
and $\delta_u$ is the point mass at $u$. We then have 
\[ \mu^y(du)Q(u,dv) = \mu^y(dv)Q(v_k,du).\]
\end{proposition}

The proof of the proposition is similar to the ergodicity proof of the preconditioned MCMC algorithm in \cite{efendiev2006preconditioning},
and thus is omitted here.

\section{Numerical examples}\label{sec:examples}

\subsection{A signal denoising problem}
\begin{figure}
\centerline{\includegraphics[width=.5\textwidth]{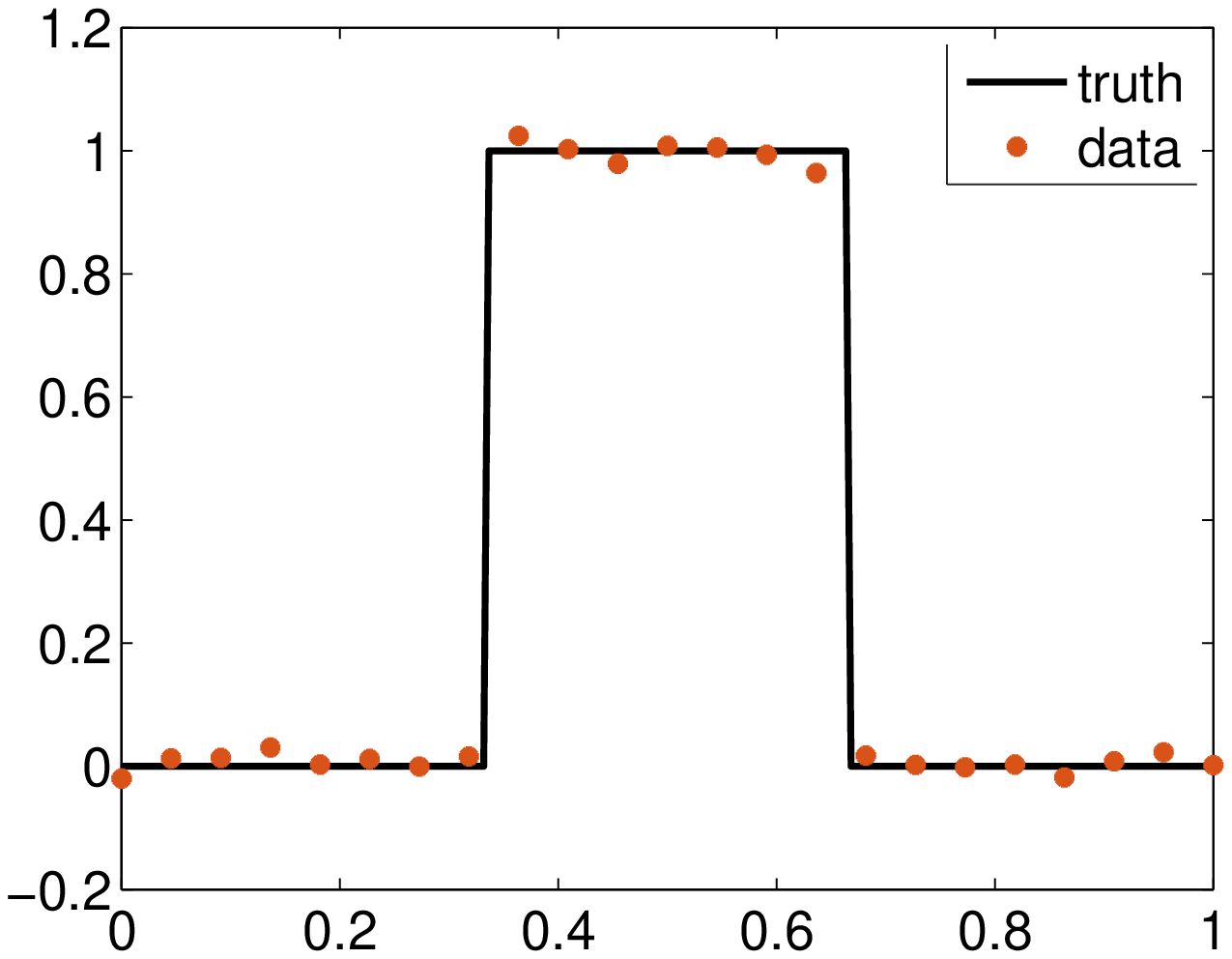}
\includegraphics[width=.5\textwidth]{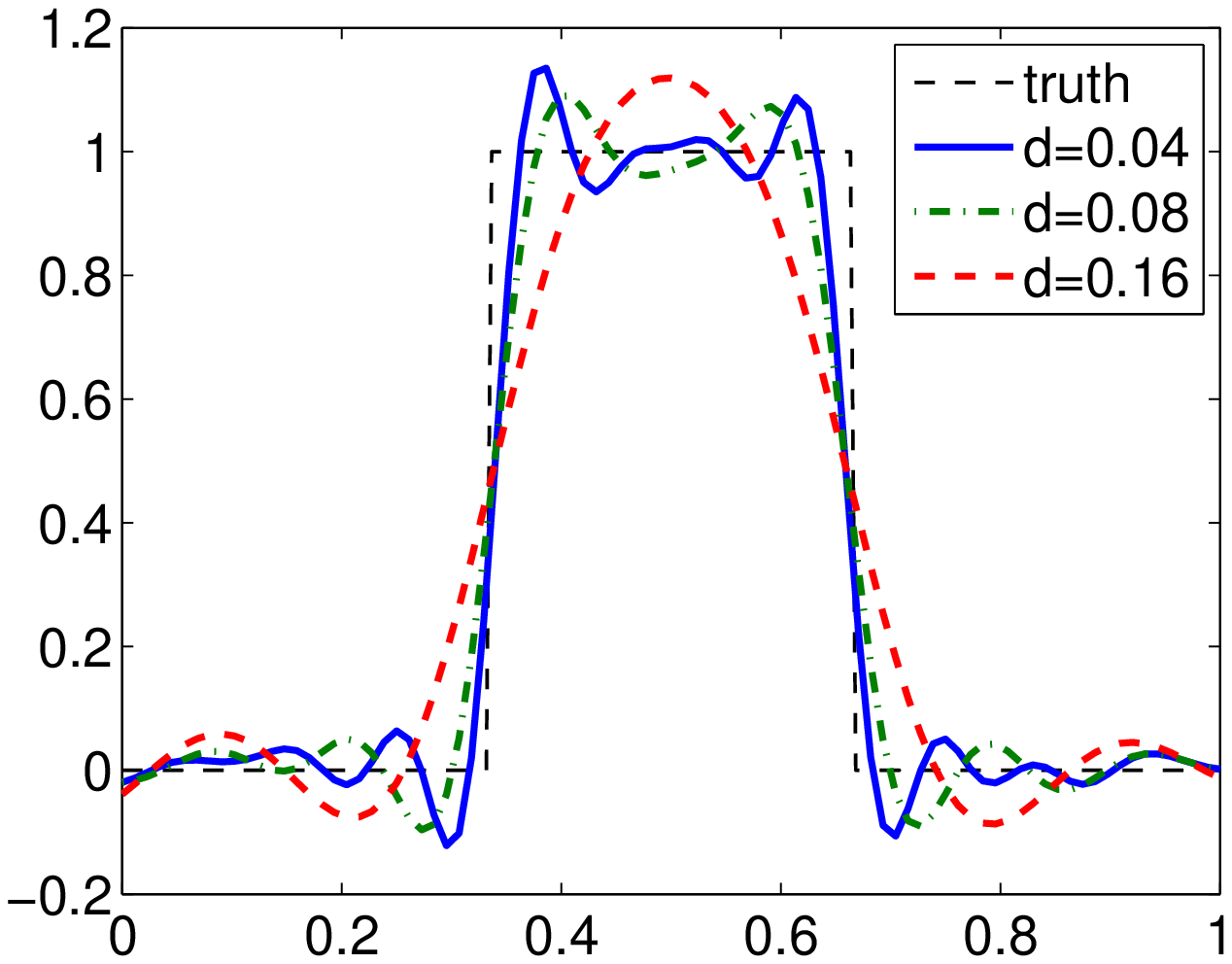}}
\caption{Left: the true signal (solid line) and the observed data points (dots).
Right: the posterior mean of the Gaussian prior with $d=0.04,\,0.08,\,0.16$.}\label{fig:signal}
\end{figure}
First we test the proposed prior on a signal denoising example. 
Namely, suppose we have a piecewise constant signal,
\[
u(t) = \bigg\{\begin{array}{cl}
0, & 0\leq t<1/3;  \\
1, & 1/3\leq t<2/3;  \\
0, & 2/3\leq t\leq1.  \end{array} 
\]
We observe $23$ data points equally distributed on $[0,\,1]$ each with independent Gaussian observation noise $N(0,0.02^2)$. 
The true signal and the data points are shown in Fig.~\ref{fig:signal} (left).
We first test the Gaussian priors,
and specifically we choose the  prior to be
 a zero-mean Gaussian measure with a squared exponential covariance:
\begin{equation}
K(t_1,t_2) = \gamma\exp\left[ -{\frac12\left(\frac{t_1-t_2}{d}\right)^2}\right].\label{e:cov}
\end{equation}
Note that, with Gaussian prior, the posterior mean can be computed analytically, 
and here we computed it with $\gamma=0.1$ and $d=0.04,\,0.08$ and $0.16$. We plot the results in Fig.~\ref{fig:signal} (right),
which clearly demonstrate that the Gaussian priors perform poorly for this function.

Next we compare the performance of our TG prior and that of the TV prior.
As has been discussed,  for the TV prior, we have to use a finite dimensional formulation, and assume the density of the prior is 
\[ p(u_N) \propto \exp(-\lambda\|u_N\|_\mathrm{TV}),\]
where the regularization parameter is taken to be $\lambda=500$.
For the TG prior, we need to specify both the TV term and the Gaussian measure. 
For the TV term we set $\lambda = 500$ and 
for the Gaussian reference measure, we assume the covariance is again given by Eq.~\eqref{e:cov}, 
with $d=0.02$ and  $\gamma=0.1$. 
With either prior, we perform the inference using three different numbers of grid points $N=89,\,177,\,353$, 
and so that we can see if the inference results depend on the dimensionality of the problem. 
As this example does not involve computationally demanding forward model, we choose to sample the posterior 
with the standard pCN algorithm. 
For the TV prior, we draw $10^{8}$ samples for the cases $N=89$ and $N=177$ and $5\times 10^{8}$ samples for $N=353$.
We use such large numbers of samples to ensure reliable estimates of the inference. 
We plot the posterior mean of the TV prior in Fig.~\ref{fig:cm} (left), and we can
see (especially from the zoomed-in plots) that the results of the different numbers of grid points depart evidently from each other,
which indicates that the inference results of the TV prior depends on the discretization dimensionality.
These results are consistent with the findings reported in~\cite{lassas2009discretization}. 
Next we draw $10^{8}$ samples from the posterior with the TG prior, for all the three numbers of grid points,
 and plot the resulting posterior means in Fig.~\ref{fig:cm} (right). 
 We can see that the results for the three different grid numbers look almost identical, suggesting 
 that the results with the TG prior are independent of discretization dimensionality. 
 We also can see that, compared to the standard Gaussian priors,  the TG prior can much better detect the sharp jumps of the signal, thanks to the presence of the TV term.
 As is mentioned in the introduction, a major advantage of the Bayesian method is that it can quantify the uncertainty in the estimates. 
 To show this, in Fig.~\ref{fig:ci-samp} (left), we plot the $95\%$ pointwise credible interval (CI) of unknown,
 and in Fig.~\ref{fig:ci-samp}  we plots $10$ random samples drawn from the prior and the posterior.  
 These plots demonstrate the difference in the Bayesian and the deterministic methods for solving inverse problems. 
 
Finally we want to exam how the inference results depend on the regularization parameters, i.e., $\lambda$ and $\gamma$, for the TG prior. 
Specifically we perform the inferences with different values of $\lambda$ and $\gamma$, each with $10^8$ samples,
and show the results in Fig~\ref{fig:hp}.  
 In Fig.~\ref{fig:hp} (left), we show the posterior means computed with $\gamma =0.1$ and $\lambda = 300,\,500,\,700$;
 in Fig.~\ref{fig:hp} (right), we show the posterior means computed with $\lambda=500$ and 
 $\gamma=0.05,\,0.1,\,0.2$.
 Both figures suggest that the inference results are rather robust with respect to the values of these parameters. 

Note, however, that this is only a simply toy problem, and in practice, the problems can be much more difficult: 
the observation noise can be much stronger and the forward operator can be much more ill-posed, etc. 
Thus, to further evaluate the performance of the TG prior, we consider a more complicated problem in the next example. 

\begin{figure}
\centerline{\includegraphics[width=.5\textwidth]{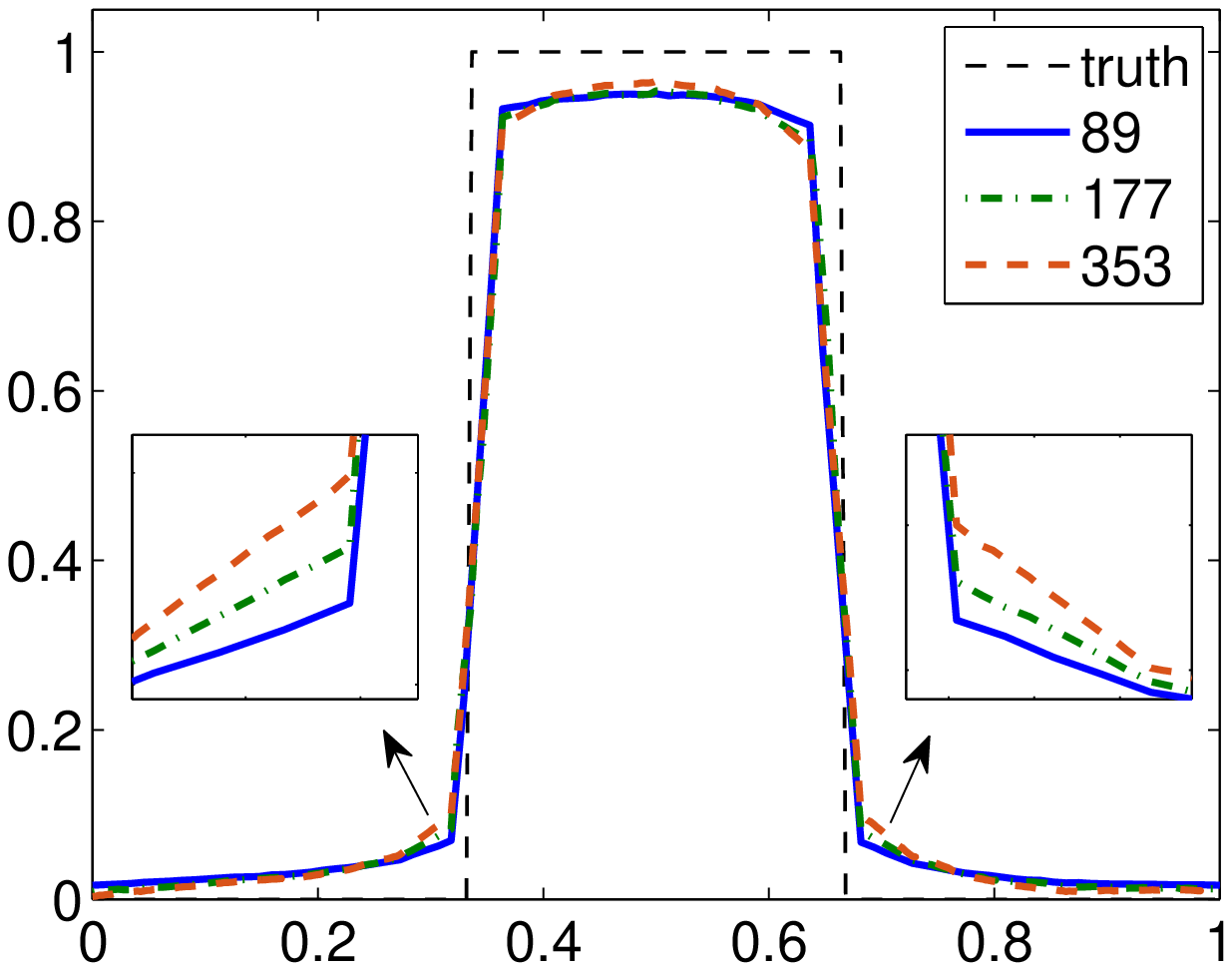} 
\includegraphics[width=.5\textwidth]{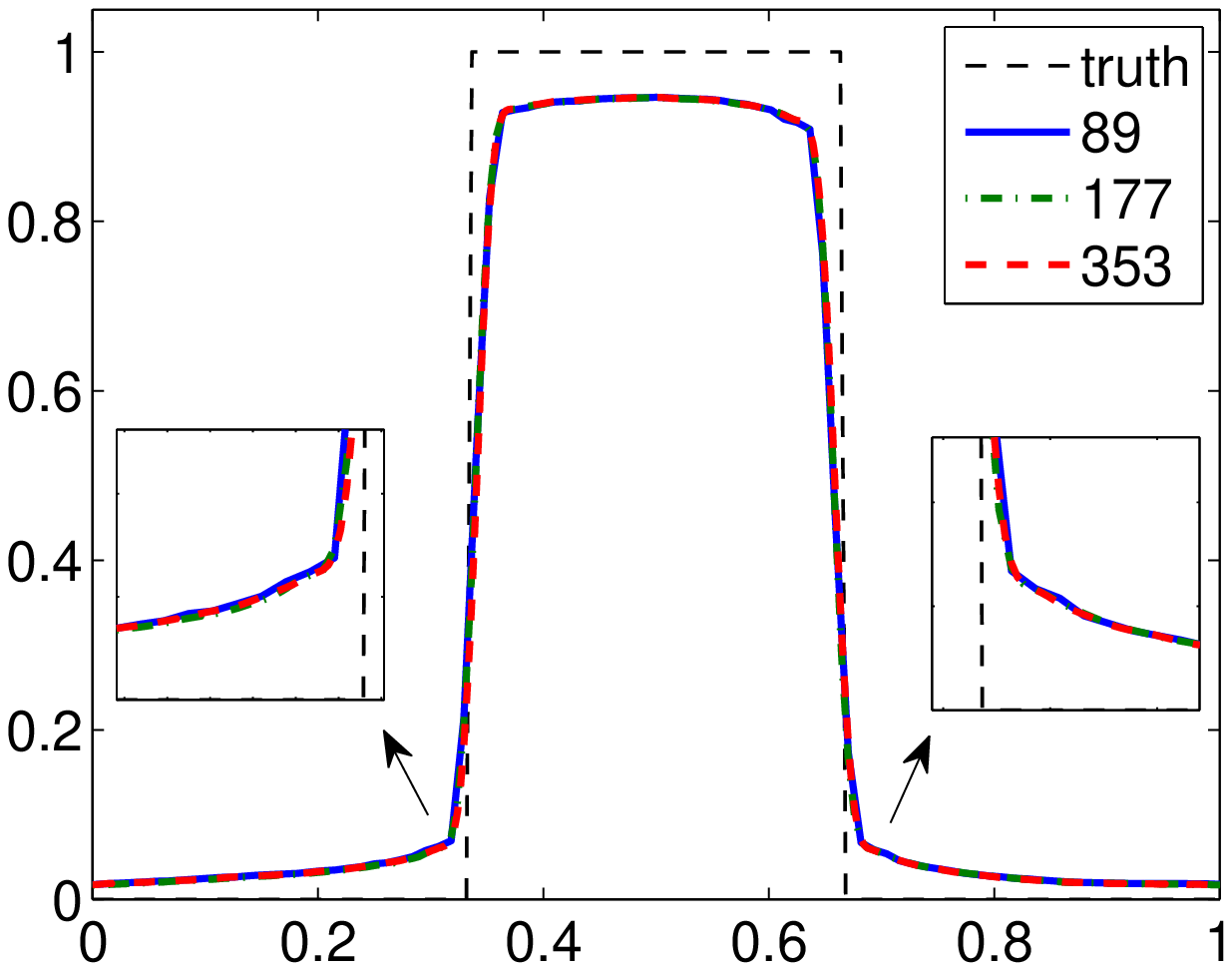}} 
\caption{Left: the posterior mean results with the TV prior, computed with grid point numbers $N=89,\,177,\,353$.
Right: the posterior mean results with the TG prior, computed with grid point numbers $N=89,\,177,\,353$.
In both figures, the insets show the zoom-in view near the jumps. }\label{fig:cm}
\end{figure}

\begin{figure}
\centerline{\includegraphics[width=.5\textwidth]{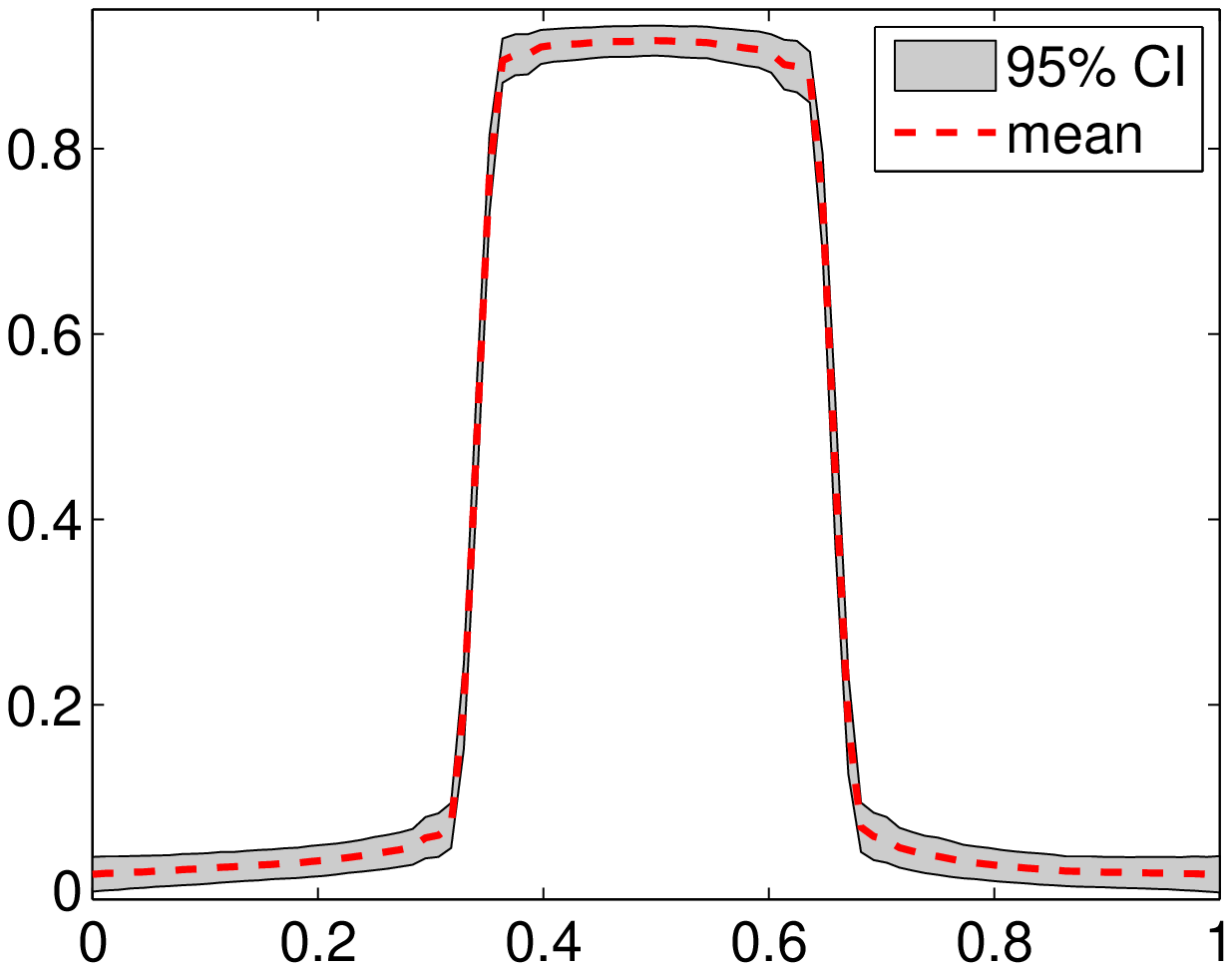} 
\includegraphics[width=.5\textwidth]{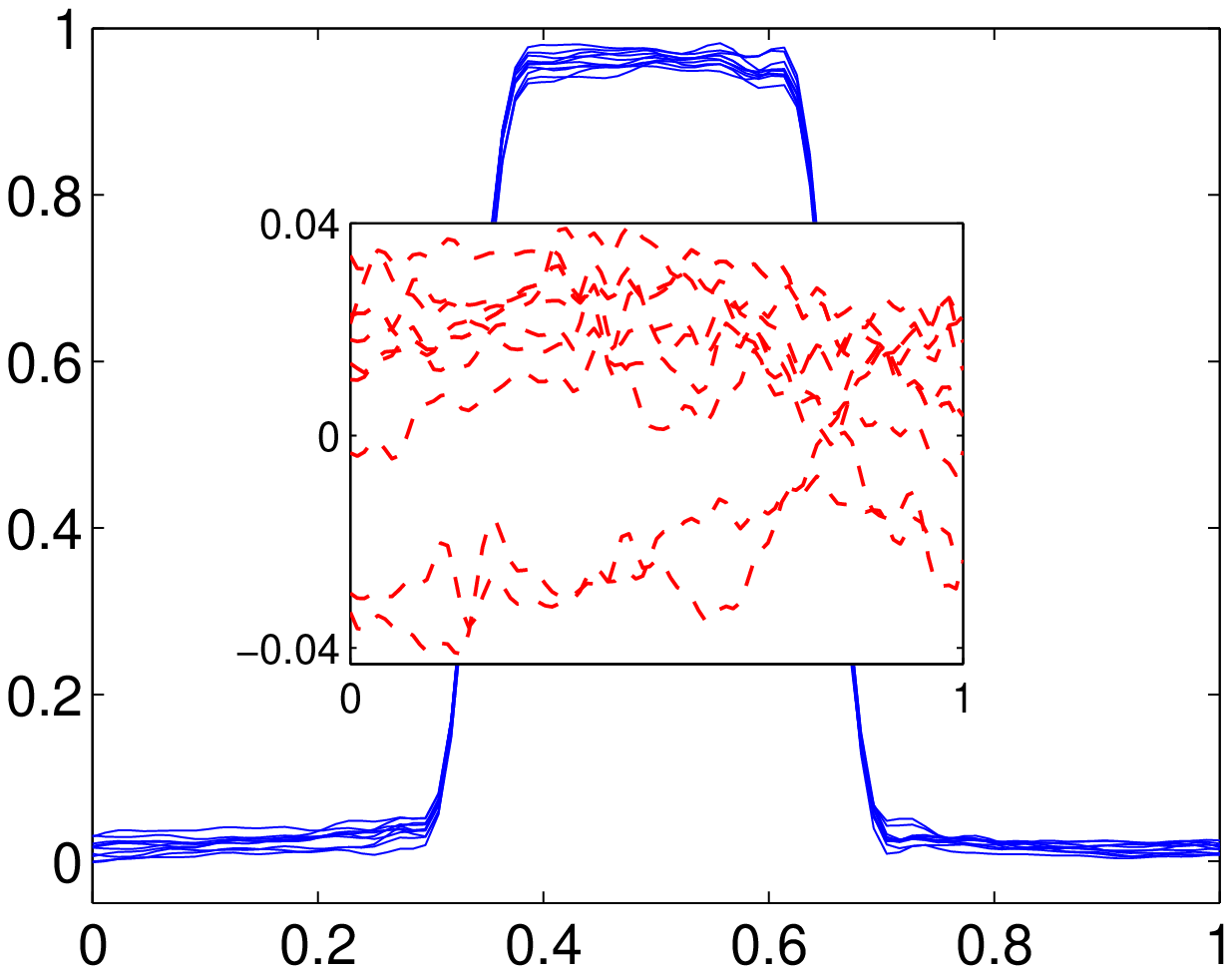}} 
\caption{Left: the pointwise $95\%$ credible interval (CI).
Right: 10 random samples drawn from the the posterior (solid lines) and from the prior (inset, dashed lines). }\label{fig:ci-samp}
\end{figure}

\begin{figure}
\centerline{\includegraphics[width=.5\textwidth]{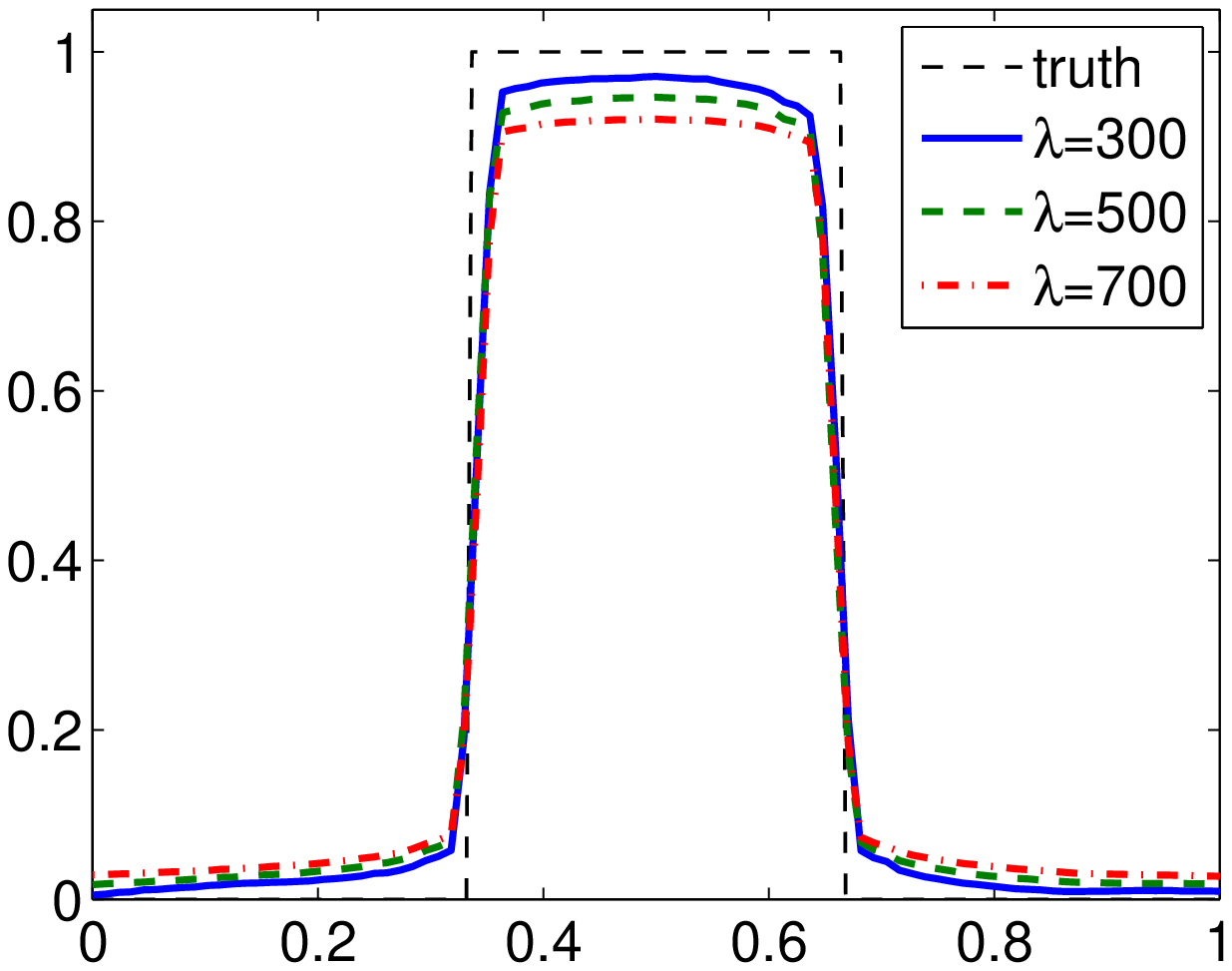} 
\includegraphics[width=.5\textwidth]{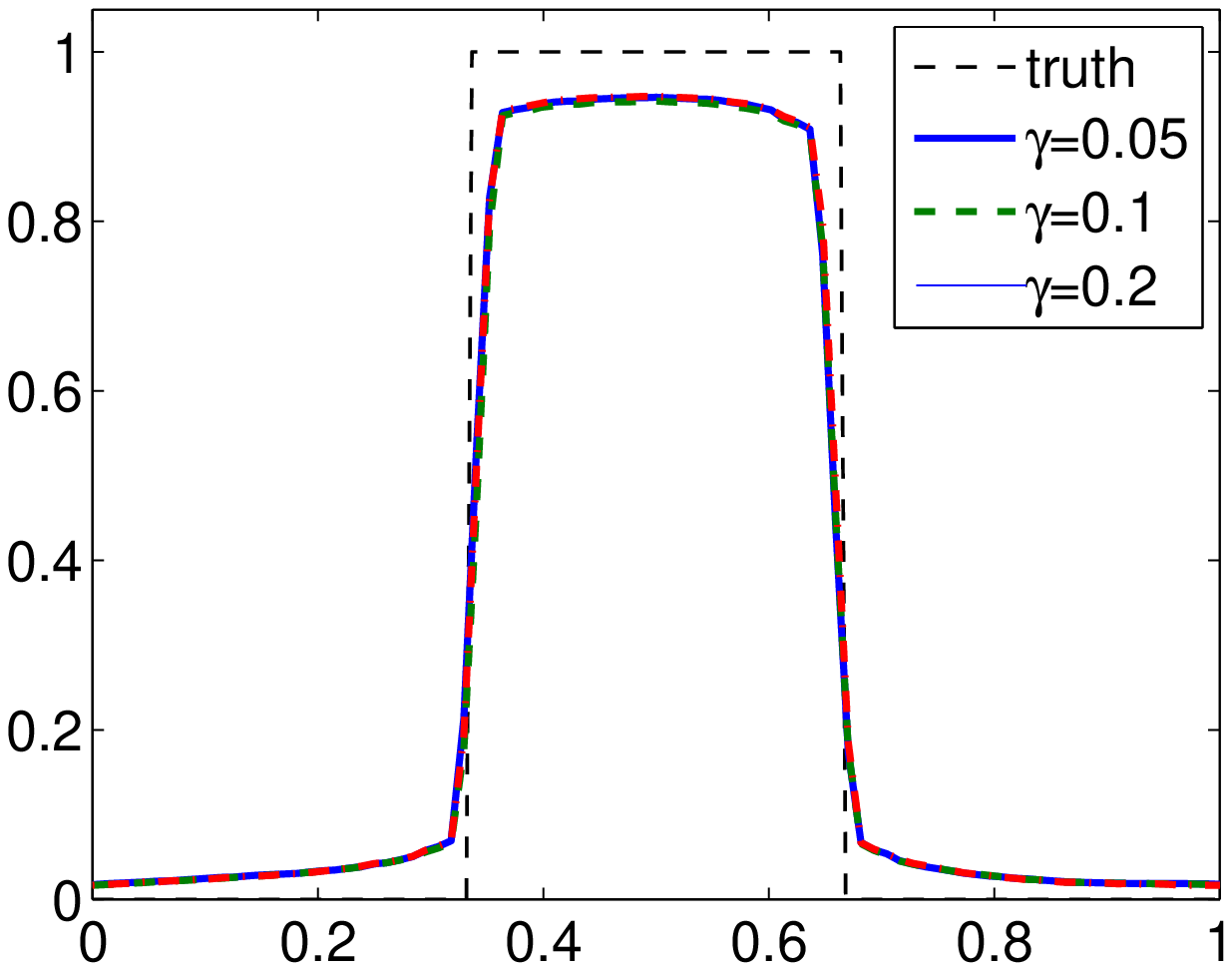}} 
\caption{The posterior means of the TG prior, computed with various values of $\gamma$ and $\lambda$. Left: the means computed with $\gamma =0.1$ and $\lambda =300,\,500,\,700$.
Right: the means computed with $\lambda=500$ and  $\gamma=0.05,\,0.1,\,0.2$.}\label{fig:hp}
\end{figure}

\subsection{Estimating the Robin coefficients}
Here we consider the one dimensional heat conduction equation in the region $x\in [0,L]$,
\begin{subequations}
\label{e:heat}
\begin{align}
&\frac{\partial u}{\partial t}(x,t) = \frac{\partial^2 u}{\partial x^2}(x,t), \\
&u(x,0)=g(x), 
\end{align}
with the following Robin boundary conditions:
\begin{align}
&-\frac{\partial u}{\partial x}(0,t) + \rho(t) u(0,t) = h_0(t),\\
 &-\frac{\partial u}{\partial x}(L,t) + \rho(t) u(L,t) = h_1(t).
\end{align}
\end{subequations}
Suppose the functions $g(x)$,  $h_0(t)$ and $h_1(t)$ are all known, and we want to estimate the unknown Robin coefficient $\rho(t)$
from certain measurements of the temperature $u(x,t)$. 
The Robin coefficient $\rho(t)$ characterizes thermal properties of the conductive medium on the interface 
which in turn provides information on certain physical processes near the boundary, 
and such problems have been extensively studied in the inverse problem literature, e.g, ~\cite{jin2012numerical,yang2009identification}. 
 
\begin{figure}
\centerline{\includegraphics[width=.5\textwidth]{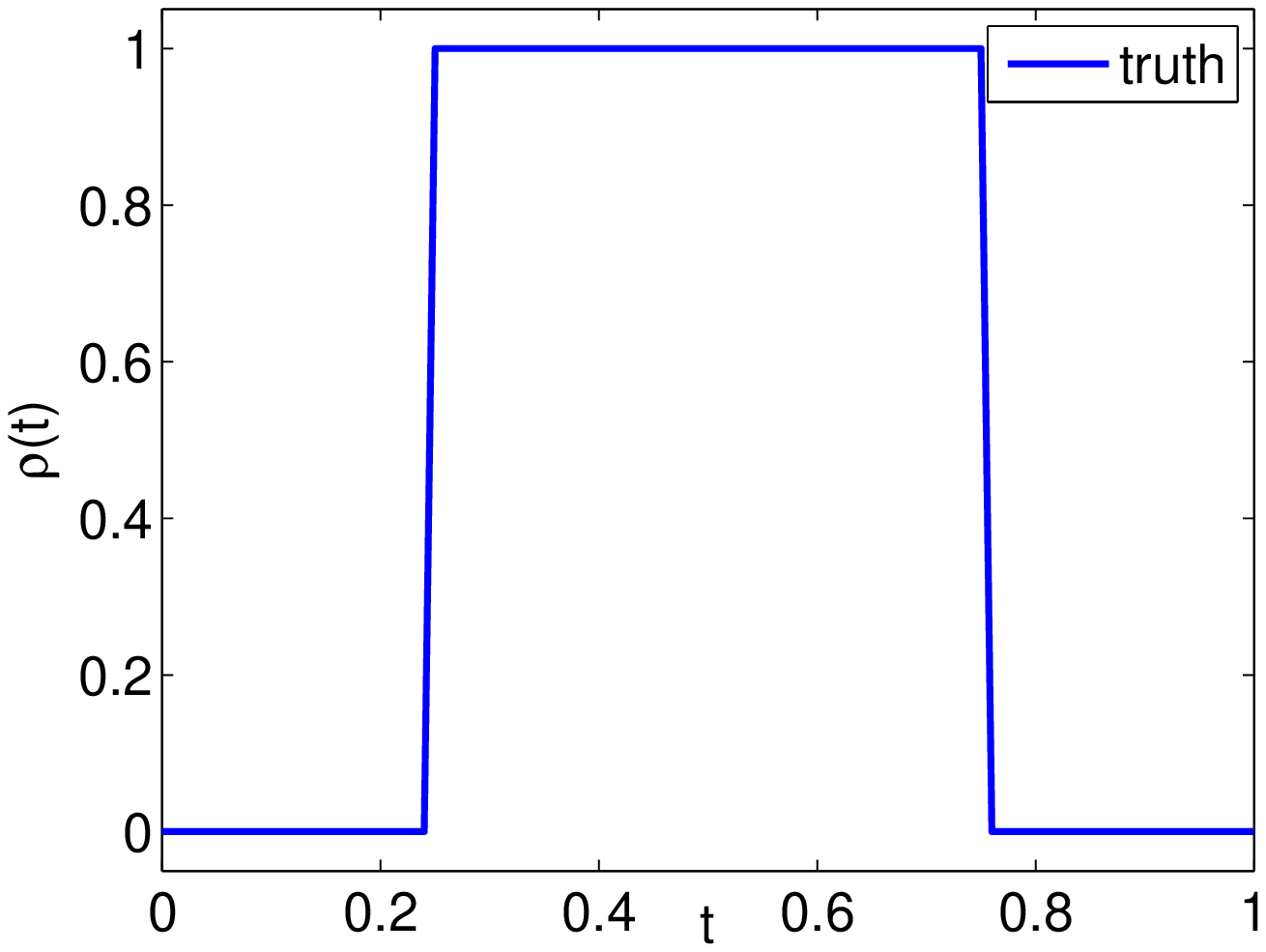}\includegraphics[width=.5\textwidth]{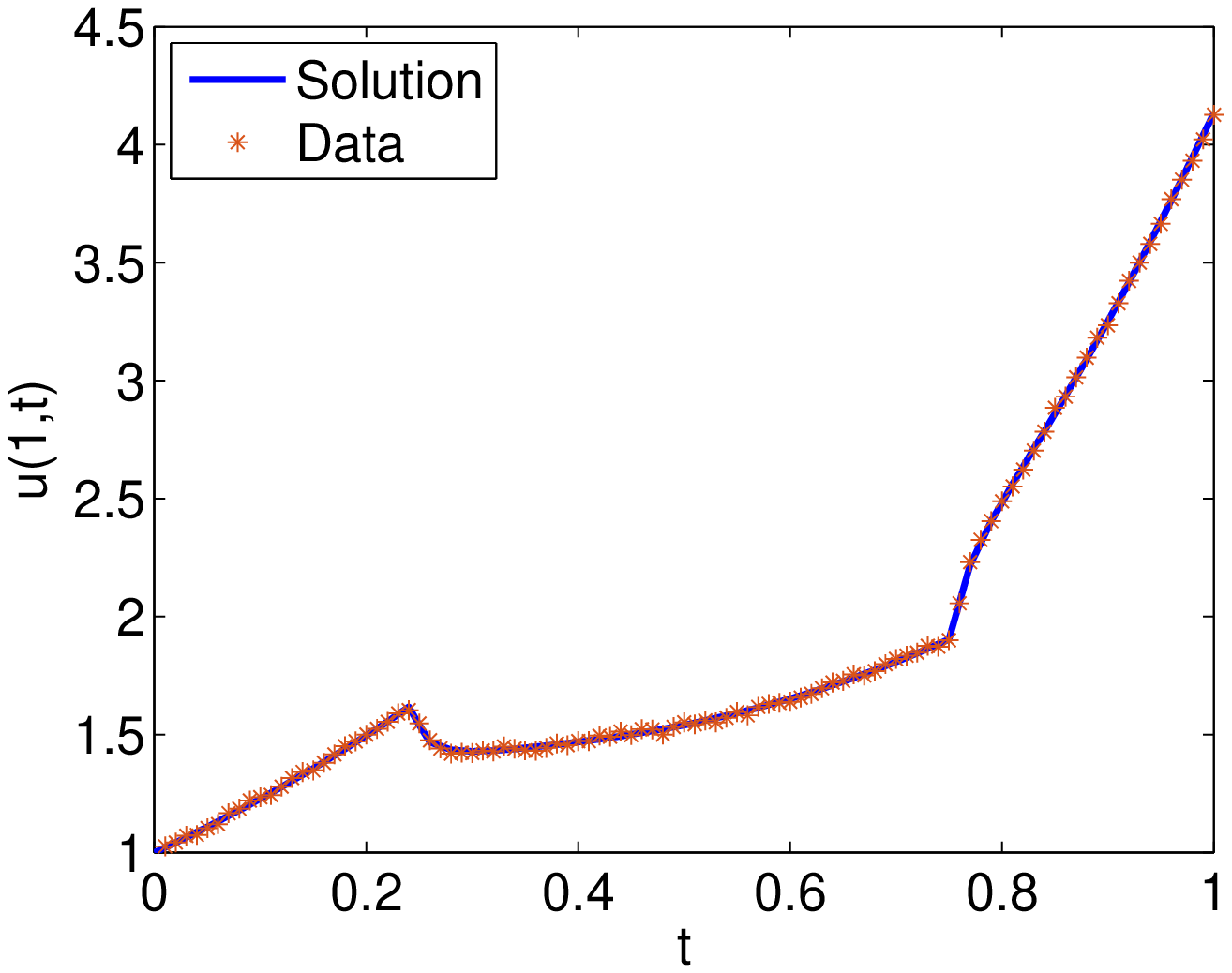}}
\caption{Left: the true Robin coefficient. Right: the true solution at $x=1$ (solid line) and the observed data points (crosses).}\label{fig:robin}
\end{figure}

\begin{figure}
\centerline{\includegraphics[width=.5\textwidth]{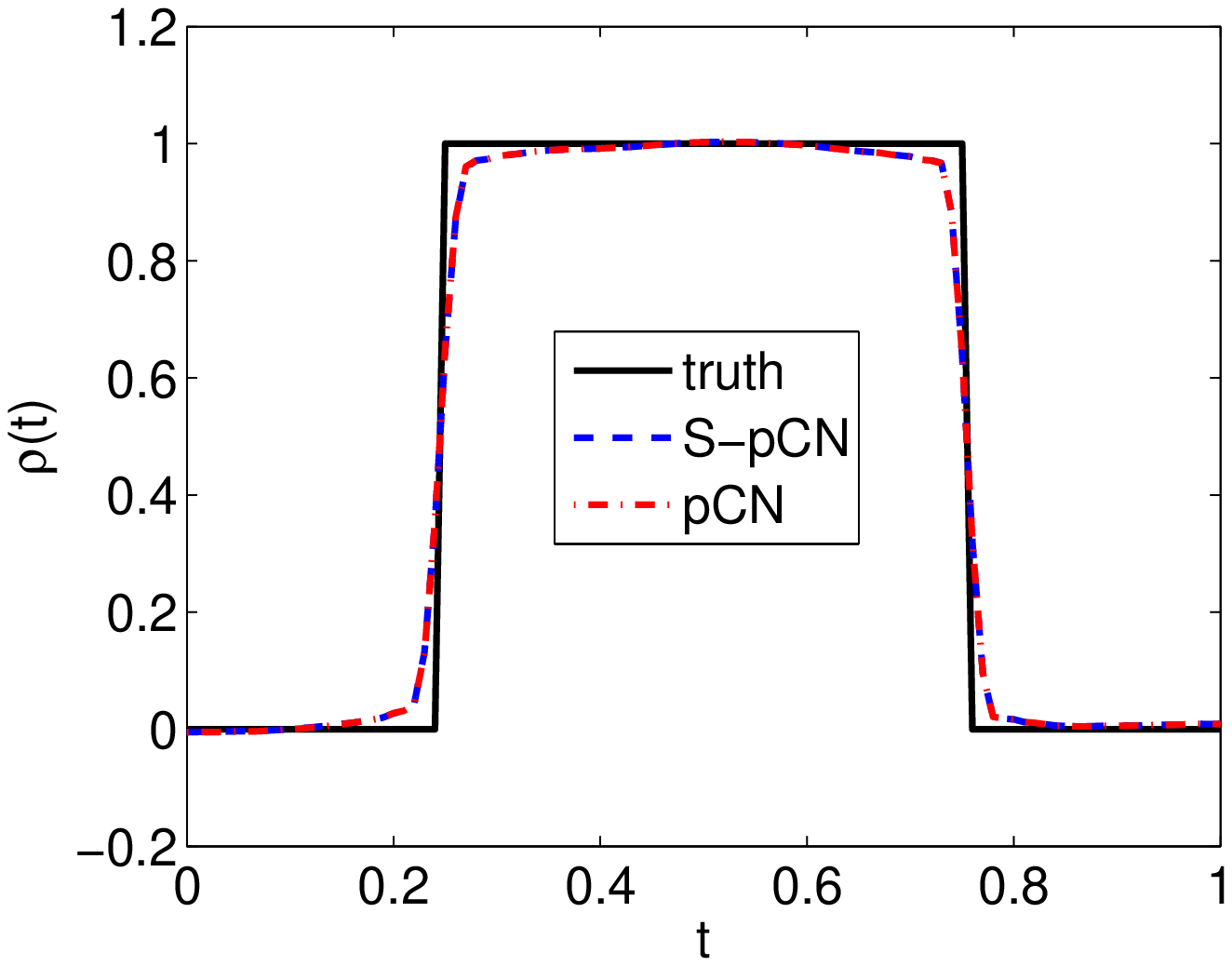}\includegraphics[width=.5\textwidth]{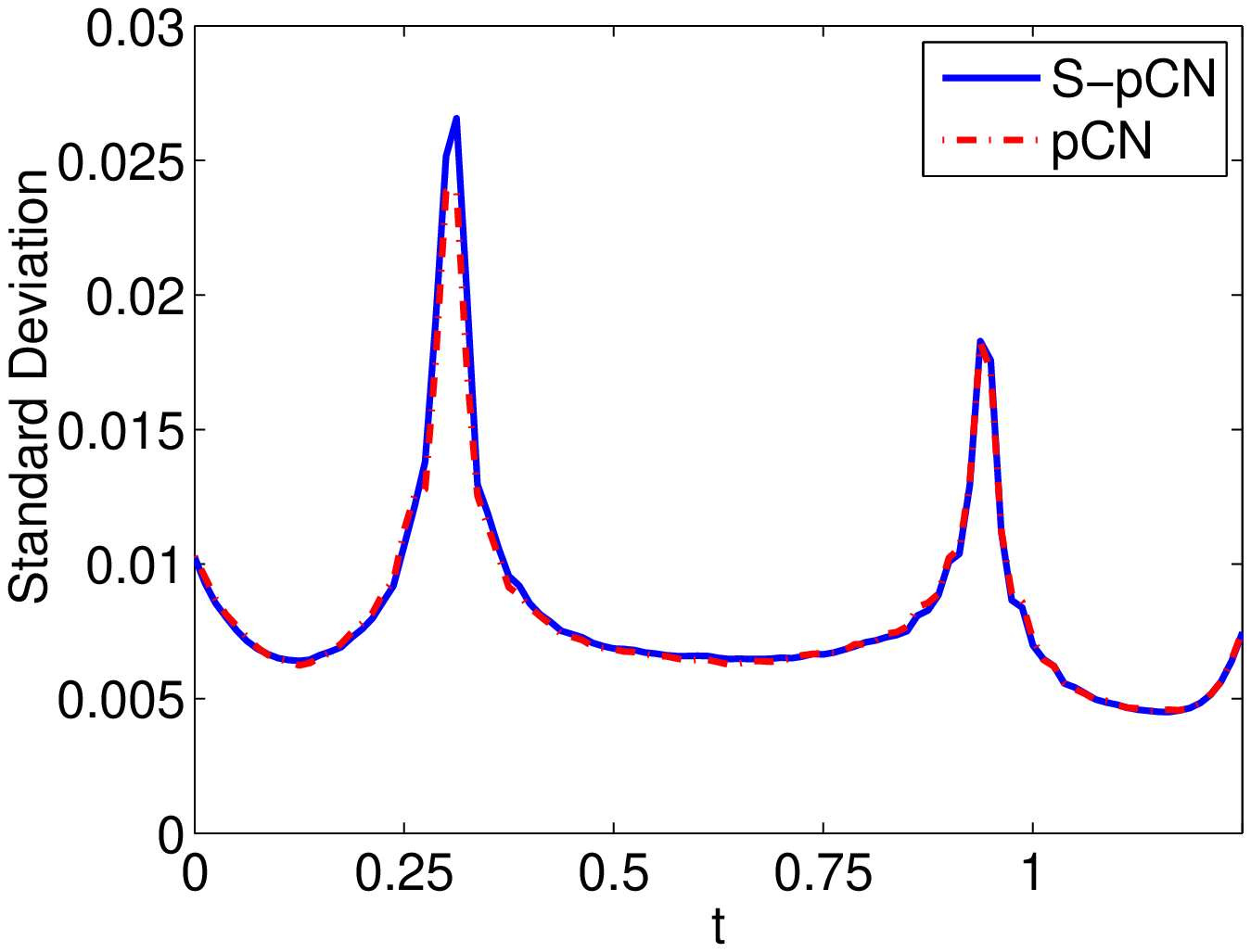}}
\caption{Left: the posterior mean computed by the S-pCN (dashed line) and by the standard pCN (dashed-dotted line), compared to the truth (solid line). Right: the posterior standard deviation computed with S-pCN (dashed) and with pCN (dashed-dotted).}\label{fig:mean-var}
\end{figure}
To be specific, we assume that a temperature sensor is placed at one boundary $x=L$ of the medium, 
and we want to estimate $\rho(t)$ in the time interval $[0,\,T]$ from temperature records measured by the sensor during $[0,\,T]$. 
The resulting forward operator $G$ satisfies the assumptions A.1, which can be derived from the theoretical results provided in \cite{jin2012numerical}.
In this example we choose $L=1$, $T=1$ and the functions to be
\[g(x)=x^2+1,\quad  h_0 =t(2t+1),\quad h_1=2+t(2t+2).\]
 Moreover, the temperature is measured 100 times (equally spaced) in the given time interval and
the error in each measurement is assumed to be an independent zero-mean Gaussian random  variable  with  variance  $0.01^2$. 
The ``true'' Robin coefficient is a piece-wise constant function shown in Fig.~\ref{fig:robin} (left), 
and the data is generated by substituting the true Robin coefficient into Eqs.~\eqref{e:heat}, solving the equation and adding noise 
to the resulting solution, where both the clean solution and the noisy data are shown in Fig.~\ref{fig:robin}~(right).  

We now perform the Bayesian inference with the TG prior. In particular, we choose $\lambda=300$ in the TV term, 
and for the Gaussian measure, the covariance is again given by Eq.~\eqref{e:cov} with $\gamma =0.1$ and $d=0.02$. 
Moreover, the number of grid points is taken to be $200$ in this example, and the equation~\eqref{e:heat} is solved with the finite difference scheme used in \cite{yang2009identification}.  
We sample the posterior with the standard pCN and our S-pCN algorithms, and with either method, we draw $10^{6}$ samples
from the posterior with additional $0.5\times 10^{6}$ samples are used in the burn-in period. 
We set $\beta = 0.02$ in both algorithms,  and in the S-pCN method we choose $k=10$, i.e. 10 iterations being performed in the first stage. 
As a result the average acceptance rate in the pCN scheme is about $15\%$
and that in the S-pCN scheme is about $40\%$. 

In Fig.~\ref{fig:mean-var}, we show the posterior mean (left) and standard deviation (right) computed by both methods, 
and we can see that the results of both methods agree rather well with each other, suggesting that both methods can correctly generate samples from the 
posterior distribution. 
We want to note that the posterior mean resulting from the TG prior can reasonably detect the jumps in the Robin coefficient; 
we have not optimized the hyperparameters in the prior and  the results may be further improved if we do so. 
We now compare the efficiency performance of the two methods. 
First, we show in Fig.~\ref{fig:trace} the trace plots of the two methods for the unknown $\rho(t)$ at $t=0.2,\,0.5$ and $0.8$.
The trace plots provide a simple way to examine the convergence behavior of the MCMC algorithm. Long-term trends (i.e., low mixing rate) in the plot indicate that successive iterations are highly correlated and that the series of iterations have not converged. 
In this regard, the trace plots indicate that the chains produced by the S-pCN method mix much faster than
those by the standard pCN. 
Next we examine the autocorrelation functions (ACF) of the chains generated by both methods.
Once again we consider the points at  $t=0.2,\,0.5$ and $0.8$ and we plot the ACF for all the three points in Fig.~\ref{fig:acf}. 
One can see from the figure that, for all three points, the ACF of the chain generated by the S-pCN decreases much faster than that of the standard pCN,
suggesting that the S-pCN method achieves a significantly better performance.
Alternatively, we look at the ACF of lag $100$ at all the grid points, which is plotted in Fig.~\ref{fig:ess-lag100} (left),  
and we can see that, the ACF of the chain generated by the S-pCN is much lower than that of the standard pCN. 
The effective sample size (ESS) is another common measure of the sampling efficiency of MCMC~\cite{Kass1998}. 
ESS is computed 
by \[\mathrm{ESS} = \frac{N}{1+2\tau},\]
where $\tau$ is the integrated autocorrelation time and $N$ is the total sample size, and it gives an estimate of the number of effectively independent draws in the chain.
We computed the ESS of the unknown $\rho$ at each grid point and show the results in Fig.~\ref{fig:ess-lag100} (right).
The results show that the S-pCN algorithm on average produces around 5 times more effectively independent samples than the standard pCN.

\begin{figure}
\centerline{\includegraphics[width=1.0\textwidth]{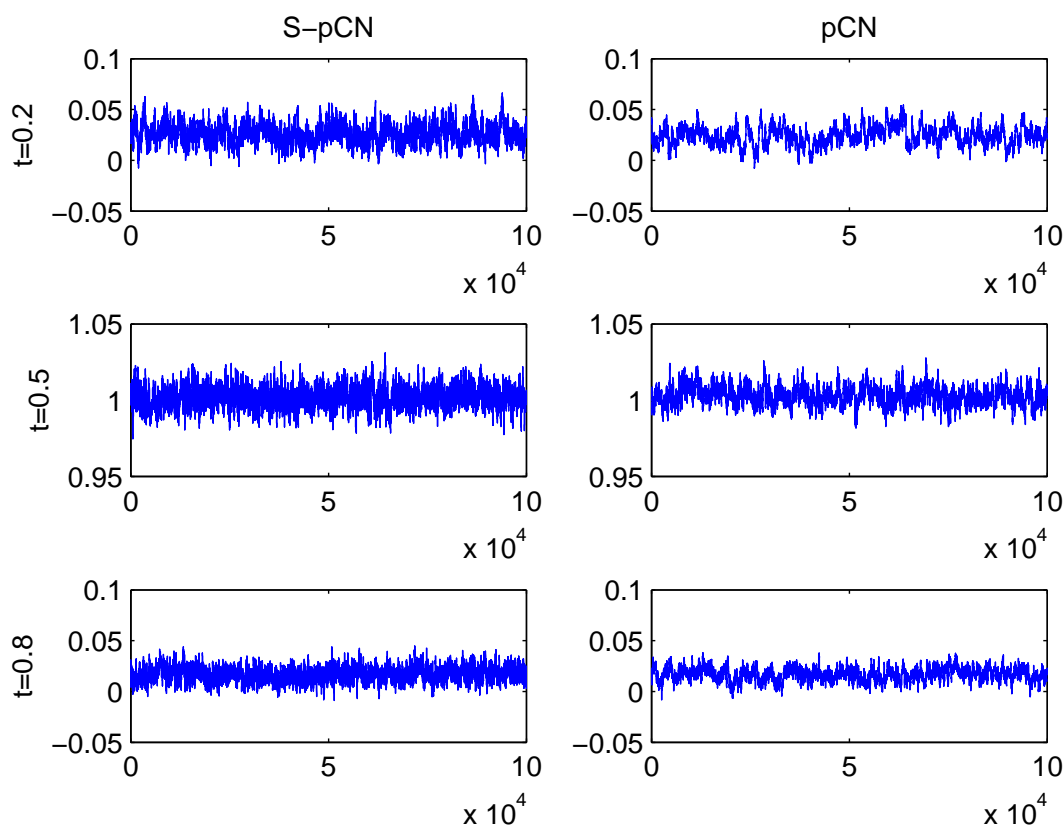}}
\caption{The trace plots for the points at $t=0.2,\,0.5$, and $0.8$ with the S-pCN (left) and the pCN (right) methods.}\label{fig:trace}
\end{figure}


\begin{figure}
\centerline{\includegraphics[width=1\textwidth]{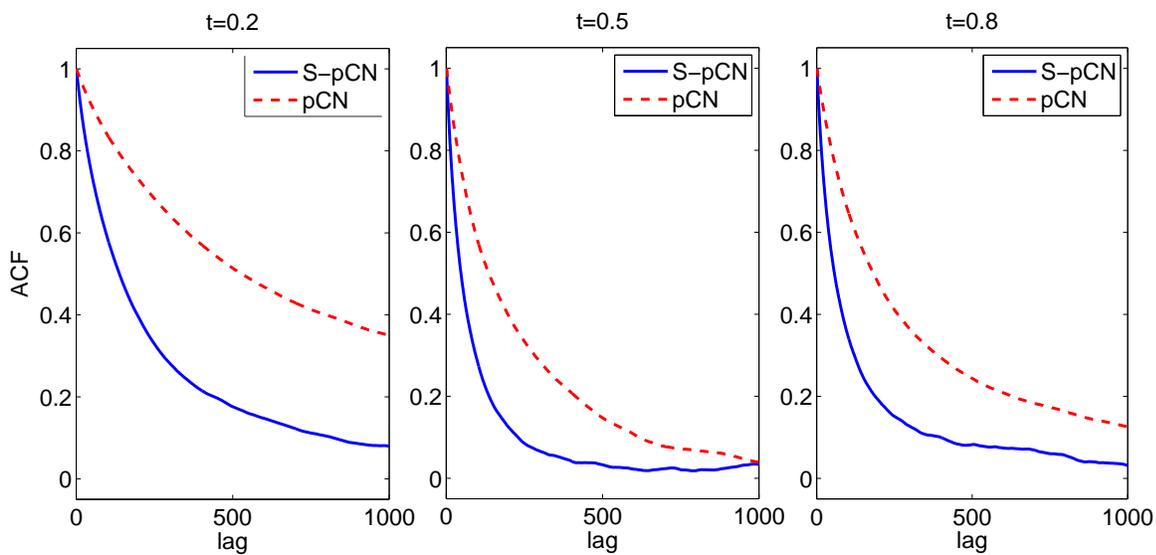}}
\caption{The ACF of $\rho$ at $t=0.2,\,0.5,\,0.8$.}\label{fig:acf}
\end{figure}

\begin{figure}
\centerline{\includegraphics[width=.5\textwidth]{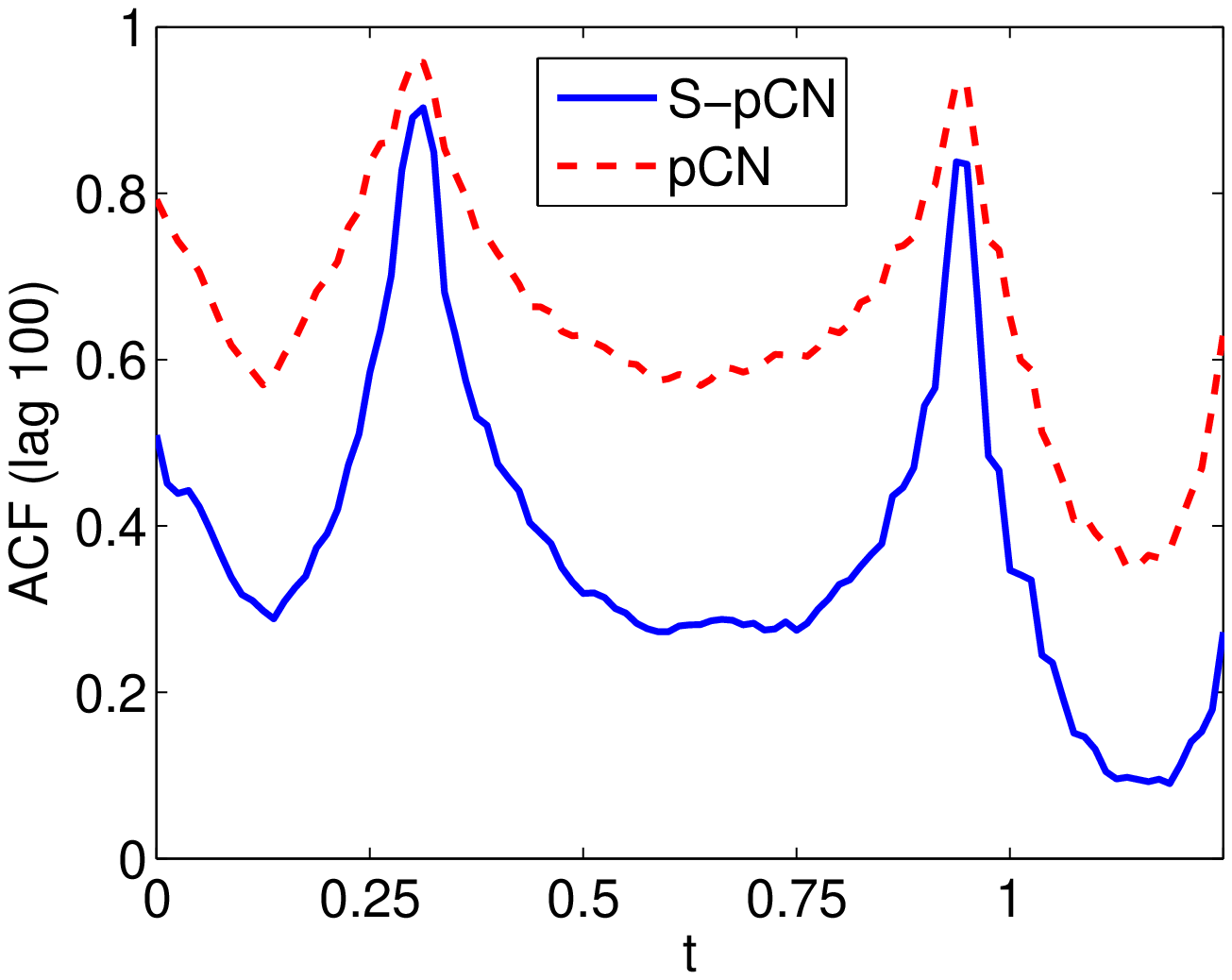}\includegraphics[width=.5\textwidth]{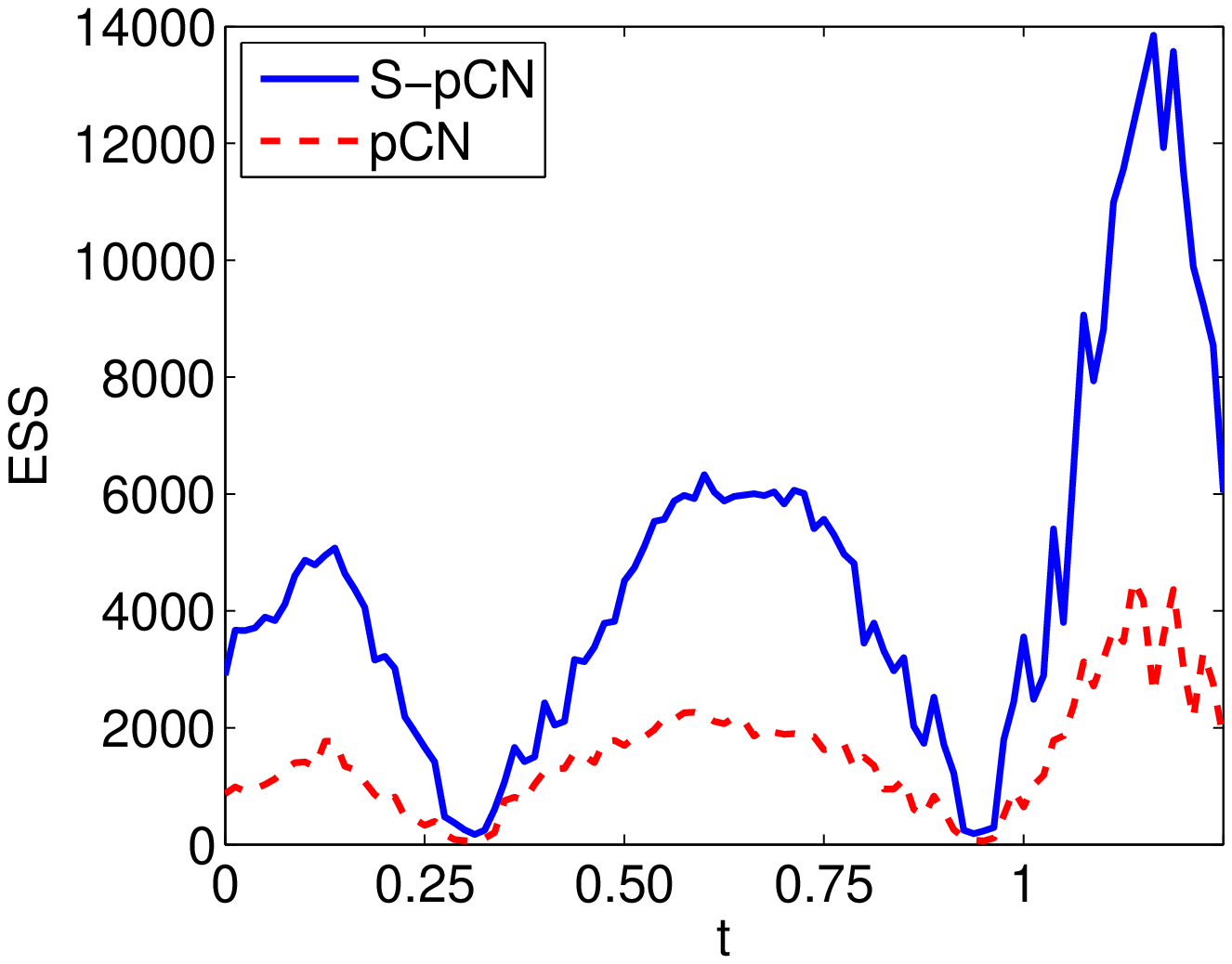}}
\caption{Left: the ACF of lag 100, at each grid point. Right: the ESS at each grid point.}\label{fig:ess-lag100}
\end{figure}

\section{Conclusions}\label{sec:conclusions}
In summary, we have presented a TV-Gaussian prior for infinite dimensional Bayesian inverse problems.
We use the TV term to improve the ability to detect jumps and use the Gaussian reference measure to ensure 
that it results in a well defined the posterior measure.  
Moreover, we show that the resulting posterior distributions  depend continuously on data and more importantly
can be well approximated by finite dimensional representations. 
We also present the S-pCN algorithm which can significantly improve the sampling efficiency by 
simply splitting the standard pCN iterations into two stages.
Finally we provide some numerical examples to demonstrate the performance of the TG prior and the efficiency of the S-pCN algorithm. 
We believe the proposed TG prior can be useful in many practical inverse problems involving functions with sharp jumps.

Several extension of the present work are possible. 
First, it is interesting to consider the connection between the proposed TG prior and the EN type regularization in deterministic inverse problems.
As is shown, the MAP estimator of the TG prior yields the solution of the deterministic inverse problem with an EN type of regularization.
We believe such a connection can provide some interesting theoretical results of the TG prior, the investigation of which is of our interest. 
Secondly, it should also be noted that, throughout the work we assume the regularization parameter $\lambda$ is given, while in practice,
determining $\lambda$ can be a highly nontrivial task. 
A simply solution here is to determine $\lambda$ with the techniques used in the deterministic setting, and then use the result directly in the Bayesian inference. Nevertheless, developing rigorous and effective methods to determine the regularization parameter in the Bayesian setting is a problem of significance. 
A possible solution is to impose a prior on $\lambda$ and formulate a hierarchical Bayes or empirical Bayes problem to estimate both $\lambda$ and $u$.
Finally, a very natural extension of the present work is to consider other choices of regularization term $R$, to reflect different prior information. 
We plan to investigate these issues in the future. 
\section*{Acknowledgment}
 We wish to thank the anonymous reviewers for their very helpful suggestions. 
The work was partially supported by NSFC under grant number 11301337. 
ZY and ZH contribute equally to the work.

\section*{References}
\bibliographystyle{plain}
\bibliography{prior}

\end{document}